\documentclass[12pt,a4paper,final]{iopart}
\usepackage{iopams}

\usepackage{amsthm}
\usepackage{setstack}
\usepackage{enumerate}

\usepackage[normalem]{ulem}
\usepackage{color}
\definecolor{blue}{rgb}{0,0,1}
\definecolor{red}{rgb}{1,0,0}

\theoremstyle{plain}
\newtheorem{theorem}{Theorem}
\newtheorem{lemma}[theorem]{Lemma}

\newtheorem{corollary}[theorem]{Corollary}
\newtheorem{definition}[theorem]{Definition}
\newtheorem{remark}[theorem]{Remark}
\newtheorem{example}[theorem]{Example}

\newcommand{\I}{\mathrm{i}}
\newcommand{\R}{\mathbb{R}}
\newcommand{\N}{\mathbb{N}}

\newcommand{\C}{\mathbb{C}}
\newcommand{\de}{\mathrm{d}}
\newcommand{\KK}{\mathtt{K}}
\newcommand{\supp}{\mathop{\mathrm{supp}}\nolimits}

\newcommand{\Ker}{\mathop{\mathrm{Ker}}\nolimits}
\newcommand{\EE}{\mathcal{E}}
\newcommand{\DD}{\mathcal{D}}
\newcommand{\Sch}{\mathcal{S}}
\newcommand{\Borel}{\mathrm{Borel}}
\newcommand{\Reg}{\mathcal{R}}
\newcommand{\RegS}{\mathcal{R}_{\mathcal{S}}}
\newcommand{\RegSp}{\mathcal{R}_{\mathcal{S}+}}
\newcommand{\RegD}{\mathcal{R}_{\mathcal{D}}}
\newcommand{\SchD}{\mathcal{S}_{\mathcal{D}}}
\newcommand{\Vcal}{\mathcal{V}}
\newcommand{\Rbar}{\overline{\R}}
\newcommand{\Wbar}{\underline{W}}

\newcommand{\seq}[1]{(#1_{n})_{n\in\N}}

\DeclareMathSymbol{\shortminus}{\mathbin}{AMSa}{"39}

\newcommand{\defin}[1]{\textit{\textbf{#1}}}
\newcommand{\enref}[1]{(\ref{#1})}
\newcommand{\eqref}[1]{(\ref{#1})}

\begin{document}

\title[Existence theorem on the UV limit of Wilsonian RG flows of Feynman measures]{Existence theorem on the UV limit of Wilsonian RG flows of Feynman measures}

\author{Andr\'as L\'aszl\'o\footnote{Corresponding author.}}
\address{
Dep.\ of High Energy Physics, HUN-REN Wigner Research Centre for Physics\\
Konkoly-Thege M u 29-33, 1121 Budapest, Hungary
}
\ead{laszlo.andras@wigner.hun-ren.hu}

\author{Zsigmond Tarcsay}
\address{
Dep.\ of Mathematics, Corvinus University of Budapest\\
F\H ov\'am t\'er 13-15, 1093 Budapest, Hungary\\
and Dep.\ of Appl.\ Analysis and Comp.\ Mathematics, E\"otv\"os University\\
P\'azm\'any P\'eter s\'et\'any 1/C, 1117 Budapest, Hungary}
\ead{zsigmond.tarcsay@uni-corvinus.hu}

\author{Jobst Ziebell}
\address{
Faculty of Mathematics and Computational Sciences, Friedrich-Schiller University,\\
Ernst Abbe Platz 2, 07743 Jena, Germany
}
\ead{jobst.ziebell@uni-jena.de}

\begin{abstract}
In nonperturbative formulation of Euclidean signature quantum field theory (QFT), the vacuum state is characterized by the Wilsonian renormalization group (RG) flow of Feynman measures. Such an RG flow is a family of Feynman measures on the space of ultraviolet (UV) regularized fields, linked by the Wilsonian renormalization group equation. In this paper we show that under mild conditions, a Wilsonian RG flow of Feynman measures extending to arbitrary regularization strengths has a factorization property: there exists an ultimate Feynman measure (UV limit) on the distribution sense fields, such that the regularized instances in the flow are obtained from this UV limit via taking the marginal measure against the regulator. Existence theorems on the flow and UV limit of the corresponding action functional are also discussed.
\end{abstract}

\noindent{\it Keywords}: Wilsonian renormalization, renormalization group flow, UV limit

\maketitle

\section{Introduction}
\label{S:Intro}

In the Feynman functional integral formulation of nonperturbative interacting Euclidean signature quantum field theories (QFTs), the vacuum state of a model is described by a Feynman measure living on the space of field configurations. More precisely, it is described by a Wilsonian renormalization group (RG) flow of Feynman measures on the space of ultraviolet (UV) regularized fields. As is well known, the need for the regularization and renormalization comes from the necessity of spacetime pointwise multiplication of certain distribution sense fields, as spelled out in the following.

Classically, a Euclidean field theory is described by its action functional $S$ on the space of smooth and sufficiently rapidly decreasing field configurations. It is assumed that a splitting $S=T+V$ is given, with $T$ being a quadratic positive definite kinetic term, such as a Klein--Gordon term $T(\varphi)=\frac{1}{2}\int \varphi\,(-\Delta+m^{2})\varphi$, and $V$ being a higher than quadratic degree interaction term bounded from below, such as $V(\varphi)=g\int\varphi^{4}$, the integrals performed over the spacetime.\footnote{The integration domain for the interaction term may also be restricted to a compact region of the spacetime, referred to as an infrared (IR) cutoff. In this paper, a fixed IR cutoff on the interaction term will be assumed, wherever relevant, as only the UV behavior will be investigated.} The quantization is outlined below.

It is well known, that a Gaussian Feynman measure $\gamma$ on the space of distribution sense fields can be naturally associated to the kinetic term $T$, as follows \cite{Glimm1987,Velhinho2017,Albeverio2008,Gill2008,Montaldi2017}. Denoting by $\KK$ the maximally symmetric fundamental solution (Green's distribution) to the partial differential operator $(-\Delta+m^{2})$ subordinate to $T$, the Gaussian function $Z_{\mathrm{free}}(j):=\e^{-\frac{1}{2}(j\,\vert\,\KK\star j)}$ is the Fourier transform of a probability measure $\gamma$ on the space of distribution sense fields, as implied by the well known Bochner--Minlos theorem (see e.g.\ \cite{Velhinho2017}~Corollary~1 for a concise review). This Gaussian measure $\gamma$ is a genuine probability measure, devoid of any pathologies, and is the mathematically precise definition of the Feynman measure for a non-interacting Euclidean signature QFT model. It is customary to write $\int(\dots)\,\de\gamma(\phi)$ informally as $\int(\dots)\,\e^{-T(\phi)}\,\de\phi$, the integrals performed on the distribution sense fields.

The tentative definition for the Feynman measure of the interacting QFT model is the product measure $\mu:=\e^{-V}\cdot\gamma$, meaning that $\int(\dots)\,\de\mu(\phi):=\int(\dots)\,\e^{-V(\phi)}\,\de\gamma(\phi)$, where the right hand side of the equation is the tentative rigorous definition for the informal expression $\int(\dots)\,\e^{-(T+V)(\phi)}\,\de\phi$. This construction is known to be rather obviously problematic: the measure $\gamma$ naturally lives on the space of distribution sense fields, whereas the Feynman exponent $\e^{-V}$ of interaction term lives on the space of function sense fields. Thus, the product measure $e^{-V}\cdot\gamma$ cannot be naively defined.

A workaround to this problem, called \emph{Wilsonian regularization}, was developed in the 1970s, which can be rigorously formulated as follows \cite{Wilson1975,Wilson1983,Peskin1995,Polonyi2003,Polonyi2008,Srednicki2007,Skinner2018,Dupuis2021,Kopietz2010,Bauerschmidt2019}. Consider an Euclidean QFT model encoded by a probability measure $\mu$ on the distribution sense fields, and fix a semigroup of coarse-graining operators, implemented as (in a suitable sense) $\mu$-measurable linear operators from the space of distribution sense fields to the space of function sense fields. 
The corresponding family of pushforward (marginal) measures $C_{*}\,\mu$ ($C\in\{\mbox{coarse-grainings}\}$) constitutes the renormalization group flow induced by $\mu$. By construction, the pushforward measure instances $C_{*}\,\mu$ live on the space of function sense (i.e., regularized) fields. Conversely, one may consider a family $\mu_C$ ($C\in\{\mbox{coarse-grainings}\}$) of probability measures on the space of function sense (i.e.,\ regularized) fields which are connected via pushforwards by intermediary coarse-grainings, and we shall refer to an above kind of linked family of measures as a \emph{nonterminating Wilsonian renormalization group (RG) flow}.
In this paper we investigate when such a flow admits some kind of limit measure $\mu$ from which the flow originates.\footnote{Of course, for this question to make sense, the semigroup of coarse-graining operators needs to approach arbitrarily close to the identity operation, but on continuum spacetimes this is basically automatic.}

To see that the above concept indeed corresponds to the usual constructive implementation of the Wilsonian RG, consider a free massive Klein--Gordon QFT encoded by its Gaussian measure $\gamma$ on the space of tempered distributions.
For every $\Lambda > 0$, consider the linear operator $C_{\Lambda}$ which cuts off the momentum modes $p^2 > \Lambda^2$ of a tempered distribution, whenever that operation is meaningful. (This amounts to the multiplication of a tempered distribution with a non-smooth function which is not everywhere defined.)
One can show that $C_{\Lambda}$ is indeed a well defined $\gamma$-measurable mapping in a certain sense, and that the pushforward measures $(C_{\Lambda})_{*}\, \gamma$ indeed correspond to the usual UV-cutoff versions of $\gamma$.
The next step is to include a family of Feynman exponents $e^{-V_{\Lambda}}$ ($0<\Lambda<\infty$) of interaction terms in a way that implements Wilson's idea: namely to integrate out high-momentum modes in every step.
This is precisely what the momentum cutoff operators do: the RG operation connecting $\e^{-V_{\Lambda'}}\cdot (C_{\Lambda'})_{*}\,\gamma$ with $\e^{-V_{\Lambda}}\cdot (C_{\Lambda})_{*}\,\gamma$  ($0<\Lambda'<\Lambda<\infty$) corresponds to the pushforward by $C_{\Lambda'}$, modulo an eventual field rescaling by factors $z(\Lambda')$ and $z(\Lambda)$.

Because the measurability of the sharp UV cutoff operators is intimately tied to the free theory, we will focus on the more generally applicable approach of using continuous coarse-graining operators (practically corresponding to smooth momentum cutoffs).
The benefit is that the corresponding pushforward operations are well-defined irrespective of which Borel measure they act on such that the setup generalizes to arbitrary scalar quantum field theories on either flat or curved spacetime, and without an \emph{a priori} splitting of the model to a free plus an interaction term. The precise form of the RG equation reads
\begin{eqnarray}\label{E:muzC}
 \exists\,\mathrm{real\; valued\; functional }\; z \;\mathrm{ of\; coarse{\shortminus}grainings}: \cr
 \Big. \quad \forall\,\mathrm{coarse{\shortminus}grainings }\; C,C',C'' \;\mathrm{ with }\; C''=C'\,C : \cr
 \Bigg.\qquad\quad z(C'')_{*}\,\mu_{C''} \;=\; C'_{*}\,\Big(\,z(C)_{*}\,\mu_{C}\,\Big)
\end{eqnarray}
where $C'_{*}$ denotes the pushforward by $C'$, whereas $z(C'')_{*}$ and $z(C)_{*}$ denote pushforward by the field rescaling operation by factors $z(C'')$ and $z(C)$, respectively. It is seen that the pertinent equation asserts the following: proceeding in the flow from the UV toward infrared (IR) means subsequent application of coarse-graining operators to the fields, along with a subsequent field rescaling.

In this paper we prove structural theorems on the space of nonterminating Wilsonian RG flows. We show that for scalar fields over a flat spacetime manifold, given a Wilsonian RG flow of measures $\mu_{C}$ for $C$ in a suitable class of coarse-grainings, there exists a unique measure $\mu$ from which the flow originates. Moreover, we prove that if two Wilsonian RG flows $\mu_{C}$ and $\gamma_{C}$ have an interaction potential $V_{C}$ relative to each-other such that $\mu_{C}=\e^{-V_{C}}\cdot\gamma_{C}$ holds for a suitably chosen coarse-gaining $C$, then a potential $V_{C}$ will exist for all coarse-gainings $C$, and the flow of potentials will have a UV limit $V$ on the distributional sense fields such that $\mu=\e^{-V}\cdot\gamma$ holds between the UV limit measures. It will also be shown that $V$ has the same lower bound as $V_{C}$ if those potentials were bounded from below at a suitably chosen coarse-graining $C$. We also point out possibilities for strengthening the mentioned results to manifold spacetimes. 
The UV measure existence result can be seen as the analogy of the existence theorem for the UV limit of Wilsonian RG flows of Feynman correlators in arbitrary signatures, proved in \cite{Laszlo2024}. In the present paper, however, due to the Euclidean signature setting, more constraining statements can be proved.

The shown UV theorems on nonterminating Wilsonian RG flows are also discussed in the light of certain concrete models, such as the $\varphi^{4}$ model in various dimensions, and on models with bounded potential density. With the help of the presented theorems one can demonstrate, for instance, that the limiting procedure usually applied for constructing Feynman measures of high dimensional $\varphi^{4}$ theories \cite{Frohlich2024,Barashkov2020,Barashkov2021,Aizenman2021} are themselves non-Wilsonian families. A renormalizability result on bounded potential density is also presented.

There is a difference, when comparing the described Wilsonian RG prescription to other constructive approaches in QFT, as considered e.g.\ in \cite{Rivasseau2014}. There, one consideres an already regularized instance of the model in question, living on a lattice in a finite volume.
Then, one uses sophisticated techniques such as polymer expansions to obtain bounds on partition functions or Schwinger functions which will inevitably depend on the volume as well as the lattice spacing.
The next step is to show that these bounds are sufficiently uniform to take the continuum limit.
Often, this is established utilizing the Wilsonian RG as a technique to split the integration over all fields into integration over fields with their momenta in specific intervals (momentum slices, see e.g.\ \cite{Brydges1995}).
The end result is then a family of measures $(\mu_\Lambda)$ indexed by UV-cutoffs $\Lambda$ which converges to some UV limit $\mu$. In other similar, but non-lattice constructions \cite{Barashkov2020,Barashkov2021} one works on the genuine continuum Euclidean spacetime, but with UV regularized fields to define a family $(\mu_\Lambda)$ convergent to some UV measure $\mu$. 
It needs to be stressed that in all these prescriptions, the family $(\mu_\Lambda)$ will in general not be Wilsonian itself, in the sense that it does not fulfill Wilsonian RG equation \eqref{E:muzC}, as will be shown by our results. What is true, on the other hand, that in certain cases, they do converge to a UV measure $\mu$ in the so-called weak topology. 
Given this UV limit $\mu$, it is therefore a natural open question to ask, what the Wilsonian RG flow induced by the UV limit $\mu$ looks like, even though $\mu$ itself may have been defined as the limit of a non-Wilsonian regularized family.
The non-Wilsonianity will be discussed in Example~\ref{Exa:runningc} in the context of $\varphi^4$ models.

For rigorous studies on RG flows of Euclidean signature theories we also refer to the construction in \cite{Ziebell2023}, which formalizes the notion of Wetterich type RG flows.
While the underlying idea, to suppress high-momentum modes inside a functional integral, remains the same, it is not implemented via pushforwards (neither sharp nor continuous).
Instead, one modifies the interaction potential $V_\Lambda$ by adding a bilinear term $B_k$ that suppresses momenta below the auxiliary scale $k \ge 0$.
Furthermore, the resulting flow deals with the \emph{effective average actions}, which are Legendre transforms of cumulant-generating functions, instead of the flow of the measures themselves.

The structure of the paper is as follows. In Section~\ref{S:Wilson} the mathematical framework for Wilsonian RG flows of regularized Feynman measures is rigorously spelled out. In Section~\ref{S:ExistenceMu} we prove the existence theorem for a UV limit Feynman measure. In Section~\ref{S:ExistenceV} we prove an existence theorem for the UV limit of the flow of relative interaction potentials between two flows. Based on the presented theorems, in Section~\ref{S:Examples} case studies are presented on concrete QFT models applying the above results. In Section~\ref{S:Conclusion} a conclusion is presented. The paper is closed by \ref{A:GaussSupp}, a review on some known important properties of Gaussian measures on the space of distribution sense fields.

The following normalization convention should be noted. The running field renormalization factor appearing in \eqref{E:muzC} may be merged into the measures $\mu_{C}$ in the flow, by considering rather the flow of measures $z(C)_{*}\,\mu_{C}$ on the space of rescaled fields (${C\in\{\mbox{coarse-grainings}\}}$). That convention, in QFT terms, corresponds to the book-keeping of the running field renormalization factor as the running coupling of the kinetic term in the action functional. This convention sets $z(C)=1$ in \eqref{E:muzC}, which is used from this point on, and therefore the running field renormalization factor will not appear in the notation explicitly in the sequel, although it is accounted for. In the pertinent convention, the factor re-appears when comparing two flows.

\section{Wilsonian renormalization group (RG) flows}
\label{S:Wilson}

The mathematical results presented in this paper focus on Euclidean field theories formulated over flat spacetimes, i.e.\ spacetimes isometric to $\R^{N}$. The sole reason for such restriction is that the proof of our main theorem (Theorem~\ref{T:ZExists}) relies on two factorization properties (Remark~\ref{R:ConvSurj}~\enref{R:ConvSurj:i} and \enref{R:ConvSurj:iii}) of convolution operators by Schwartz functions. If analogous factorization properties could be established for coarse-graining operators on generic manifolds, the arguments of the presented theorems would carry over to the manifold spacetimes as well. (For a review on coarse-graining operators on manifolds, see e.g.\ \cite{Laszlo2024}~Remark~1.) In order to simplify the notations further, we focus on the case of scalar valued bosonic fields: analogous theorems can be obtained for the case of vector valued bosonic fields, \emph{mutatis mutandis}. Fermionic fields are not considered in this paper, since they are not described by ordinary measure theory, which is our present focus.

Throughout the paper, the standard distribution theory notations are used \cite{Horvath1966,Treves1970}: $\EE$ for $\R^{N}\to\R$ smooth functions, $\Sch$ for Schwartz functions, $\DD$ for test functions, understood with their standard topologies, and $\EE'$, $\Sch'$, $\DD'$ for their continuous duals, understood with their standard strong dual topologies. In particular, $\EE'$ is the space of compactly supported distributions, $\Sch'$ is the space of tempered distributions, and $\DD'$ is the space of all distributions. Note that the spaces $\Sch$ and $\Sch'$ are only meaningful on flat spacetimes, whereas the spaces $\EE$, $\DD$ and $\EE'$, $\DD'$ are also meaningful on any orientable and oriented manifold without further assumptions. The symbol $F$ shall denote the usual $F:\,\Sch_{\C}
\to\Sch_{\C}$ Fourier transform, along with its extension $F:\,\Sch_{\C}'\to\Sch_{\C}'$ to the space of tempered distributions. Both of these are known to be topological vector space automorphisms. The normalization convention of $F$ is such that for all $\varphi,\psi\in\Sch_{\C}$ the identity $F(\varphi\star\psi)=F(\varphi)\cdot F(\psi)$ holds, where $\star$ denotes convolution and $\cdot$ denotes pointwise product. The above topological vector spaces $\EE$, $\Sch$, $\DD$, and $\EE'$, $\Sch'$, $\DD'$ are understood together with the Borel sets generated by their standard topology. For an $\eta\in\Sch$ we equip the subspace $\eta\star\Sch'\subset\Sch'$ with the subspace topology, and corresponding Borel sets. Often, we denote the convolution operator by an $\eta\in\Sch$ using the notation $C_{\eta}:\,\Sch\to\Sch$ and $C_{\eta}:\,\Sch'\to\Sch'$, furthermore we use the notation $C_{\eta}\,t=\eta\star t$ given some $t\in\Sch'$ interchangeably, according to convenience. In particular, one has $C_{\eta}[\Sch']=\eta\star\Sch'$ by notation. The standard symbols $\Rbar := \R \cup \{ \pm \infty \}$ and $\Rbar_{0}^{+}:=\R_{0}^{+}\cup\{+\infty\}$ are used for the extended real numbers and extended non-negative real numbers, along with $\e^{-\infty}:=0$ and $\e^{+\infty}:=+\infty$, furthermore we use the convention $\pm\infty\cdot 0:=0$ as customary in measure theory.

In order to motivate the mathematically precise definition of coarse-graining operators (regulators), let us consider the setting of a bare orientable and oriented manifold as a base, and fields over it, which are sections and distributional sections of some vector bundle. In such setting, coarse-graining operators are differential geometrically naturally defined as continuous linear maps from the distribution sense fields to the smooth function sense fields, such that they are properly supported (i.e.,\ they preserve compact support), are injective on compactly supported distributions, and their transpose operators likewise have the above properties, see also \cite{Laszlo2022}~Section~4 or \cite{Laszlo2024}~Section~2 for a review. From Schwartz kernel theorem it follows that these are smoothing operators by certain bivariate smooth kernel functions. When the base manifold happens to be an affine space ($\cong\R^{N}$), i.e.\ flat spacetime models are considered, one may further require translational invariance of the coarse-grainings, in which case it follows that the above operators are precisely the convolutions by some $\eta\in\DD\setminus\{0\}$. Furthermore, on $\R^{N}$ spacetimes the notions of Schwartz functions, Schwartz distributions, and Fourier transform (frequency space) become meaningful. Requiring the analogy of the above definition, but with Schwartz distributions, it follows that the translationally invariant coarse-graining operators are convolutions by some $\eta\in\Sch\setminus\{0\}$. Therefore, on flat spacetimes, the most general notion of coarse-graining operators are convolution operators $C_{\eta}$ ($\eta\in\Sch$). This definition of coarse-graining (regulator) on continuum spacetimes is widely used implicitly in the literature, see \cite{Ziebell2023,Barashkov2020,Barashkov2021,Dybalski2025} among many others as example. In other pieces of QFT literature on Wilsonian RG flows over flat spacetime, it is customary to require further properties on the allowed set of coarse-graining operators: it is customary to require the frequency spectrum $F(\eta)$ to have in addition a particular damping profile shape, as formally defined in the sequel.

\begin{definition}\label{D:Reg}
We introduce regulators with various frequency space asymptotics. As a subset of $\Sch$, let us introduce the set of \defin{regulators with Schwartz frequency tail}
\begin{eqnarray}
 \RegS := \left\{\eta\in\Sch \;\big\vert\; 0\leq F(\eta)\leq 1 \mbox{ and } \exists\,\mbox{neighbh.\ } U \mbox{of } 0\in\R^{N}:\; F(\eta)\vert^{}_{U}=1 \right\}.\cr
\end{eqnarray}
A subset of $\RegS$ is the \defin{regulators with nonvanishing Schwartz frequency tail}
\begin{eqnarray}
 \RegSp := \left\{\eta\in\RegS \;\big\vert\; 0<F(\eta) \right\}.
\end{eqnarray}
A further subset of $\RegS$ is the \defin{regulators with strictly bandlimited frequency tail}
\begin{eqnarray}
 \RegD := \left\{\eta\in\RegS \;\big\vert\; F(\eta)\in\DD \right\}.
\end{eqnarray}
\end{definition}

Throughout the paper the symbol $\Reg$ will be generally a placeholder for either of $\Sch$ or $\RegS$ or $\RegSp$ or $\RegD$ of the considered regulators, and we prove various theorems when the allowed coarse-grainings are $C_{\eta}$ ($\eta\in\Reg$) with $\Reg=\Sch$ or $\RegS$ or $\RegSp$ or $\RegD$, respectively. Note that all these sets of regulators satisfy $\forall\,\eta,\eta'\in\Reg:\;\eta\star\eta'\in\Reg$ and admit approximate identities. The customary choice in the flat spacetime QFT literature for the allowed set of regulators is $\Reg=\RegS$. We note, however, that our proofs go through when the allowed set of regulators is $\Reg=\Sch$, i.e.\ without restrictions on the Fourier spectra of the allowed regulators. The latter setting is useful for future generalization attempts of the presented theorems to manifolds, in which case no analogy to $\RegS$ or $\RegSp$ or $\RegD$ exists, due to lack of Fourier spectra. In other words, restrictions on the Fourier spectra of the regulators are avoidable in the scheme of the key proofs, and they would automatically generalize to manifolds if the suitable analogy of the factorization lemmas referred in Remark~\ref{R:ConvSurj}~\enref{R:ConvSurj:i} and \enref{R:ConvSurj:iii} could be proven for coarse-grainings on manifolds. We also note in passing that important pieces of literature on the Wilsonian RG flows over continuum flat spacetimes assume the regulators to be strictly bandlimited, i.e.\ $\Reg=\RegD$ only. This choice is not only restrictive on manifolds, but also on flat spacetimes: in the discrete analogy of the coarse-graining operation, fields are averaged in some compact vicinity of spacetime points on the spacetime mesh, being the discrete analogy of convolution by some bump function $\eta\in\DD$. Due to the Paley--Wiener--Schwartz theorem (\cite{Hormander1990}~Theorem~7.3.1), however, the Fourier spectrum of such bump function $\eta$ can never be strictly bandlimited. Nevertheless, we also study the $\Reg=\RegD$ case, i.e.\ when the allowed coarse-grainings are only the strictly bandlimited smooth frequency cutoffs, to conform to the literature.

\begin{definition}
\label{D:WRG}
Let $\Reg=\Sch$ or $\RegS$ or $\RegSp$ or $\RegD$. 
A family $(\mu_{\eta})_{\eta\in\Reg}$ of measures, such that for all $\eta\in\Reg$ the $\mu_{\eta}$ is a sigma-additive non-negative valued finite measure on the Borel sets of $C_{\eta}[\Sch']$, is called \defin{nonterminating Wilsonian renormalization group (RG) flow} whenever
\begin{eqnarray}
\label{E:WRG}
 \forall\, \eta,\eta',\eta''\in\Reg \mbox{ satisfying } C_{\eta''}=C_{\eta'}\,C_{\eta}: \qquad \mu_{\eta''} \;=\; (C_{\eta'})_{*}\, \mu_{\eta}
\end{eqnarray}
holds. Here, $()_{*}$ denotes the pushforward of a measure by a measurable map, i.e.\ by definition $\big((C_{\eta'})_{*}\,\mu_{\eta}\big)(B):=\mu_{\eta}\big(\overset{-1}{C_{\eta'}}(B)\big)$ for all Borel sets $B$, where $\overset{-1}{C_{\eta'}}(B)$ denotes the preimage of $B$ by the mapping $C_{\eta'}$.
\end{definition}

The main aim of the paper is to show that for a Wilsonian RG flow as in Definition~\ref{D:WRG}, there exists some sigma-additive non-negative valued finite measure $\mu$ on the Borel sets of $\Sch'$, such that for all $\eta\in\Reg$ the factorization identity $\mu_{\eta}=(C_{\eta})_{*}\,\mu$ holds. That would imply the family of Wilsonian regularized probability measures $(\mu_{\eta})_{\eta\in\Reg}$ having a regularization-independent UV limit.

\begin{remark}\label{R:Trace}
In the above definition, each member of the family $(\mu_\eta)_{\eta\in\Reg}$ was defined on Borel sets of $C_\eta[\Sch']$. However, as shown by Lemma~\ref{L:CetaS} stated in the sequel, when $C_\eta[\Sch']$ ($\eta\in\Sch$) is considered as a subset of $\Sch'$, it is itself seen to be a Borel measurable subset of $\Sch'$. For that reason, the Borel sigma-algebra of $C_\eta[\Sch']$ is identical to the trace sigma-algebra of Borel sets of $\Sch'$, i.e.\ 
\begin{eqnarray}
    \Borel(C_\eta[\Sch'])=\left\{B\cap C_\eta[\Sch'] \;\big\vert\; B\in \Borel(\Sch')\right\}
\end{eqnarray}
holds. This, in turn, means that every Borel measurable subset of $C_\eta[\Sch']$ is also a Borel subset of $\Sch'$. Consequently, any of the  measures $\mu_\eta$ on $C_{\eta}[\Sch']$ can be naturally extended to a Borel measure $\nu_\eta$ on $\Sch'$ as follows:
\begin{eqnarray}
     \nu_\eta(B) := \mu_\eta(B\cap C_\eta[\Sch'])\qquad B\in \Borel(\Sch').
\end{eqnarray}
The extended family of measures $(\nu_\eta)_{\eta\in\Reg}$ so obtained obviously fulfills the Wilsonian RG property \eqref{E:WRG}. Conversely, if $(\nu_\eta)_{\eta\in\Reg}$ satisfies \eqref{E:WRG}, then the family consisting of the restrictions of the $\nu_\eta$'s to Borel subsets of $C_\eta[\Sch']$, i.e.\ the measures $\mu_\eta:=\nu_\eta|_{{}_{C_\eta[\Sch']}}$, satisfy Definition~\ref{D:WRG}. Accordingly, for the purpose of Definition~\ref{D:WRG}, it does not matter whether the members of the measure family $(\mu_\eta)_{\eta\in\Reg}$ are interpreted as measures on the Borel sets of $C_\eta[\Sch']$ or as measures defined on the full Borel sigma-algebra of $\Sch'$. While the latter view is mathematically more natural and easier to handle, the former view relates more closely to the original QFT motivation of the concept of RG flow, and therefore is more appropriate for our applications carried out in Section~\ref{S:ExistenceV}. When not confusing we will not distinguish these explicitly in notation.
\end{remark}

\begin{lemma}\label{L:uX+A} 
Let $X$ and $Y$ be Hausdorff locally convex topological vector spaces, and suppose that $X$ is a countable union of weakly compact sets. Then for every continuous map $u:\,X\to Y$ and weakly closed set $A\subset Y$, the set $u[X]+A$ is a Borel set in $Y$ with respect to the weak and the original topology.
\end{lemma}
\begin{proof}
    Let $(\mathcal{K}_n)_{n\in\mathbb{N}}$ be a sequence of weakly compact sets in $X$ such that $X=\bigcup_{n\in\mathbb{N}} \mathcal{K}_n$, then 
    \begin{eqnarray}\label{E:uX+A}
        u[X]+A & \quad=\quad &  \bigcup_{n\in\mathbb N} (u[
\mathcal{K}_n]+A).    
    \end{eqnarray}
    Note that every continuous linear map $u:X\to Y$ is weakly continuous, hence $u[\mathcal{K}_n]$ is weakly compact in $Y$. Furthermore, the sum between a compact and a closed set is closed, according to
    \cite{Schaefer1999}~Theorem~1.1~(iv). Due to \eqref{E:uX+A} this implies that $u[X]+A$ is the union of  countably many weakly closed (and thus weakly Borel) sets, and hence it is weakly Borel. Since the original topology contains more open sets in comparison to the weak one, a weakly Borel set is also Borel with respect to the original topology.
\end{proof}

\begin{lemma}\label{L:CetaS}
For all $\eta\in\Sch$ the subset $C_{\eta}[\Sch']$ of $\Sch'$ is Borel measurable, with respect to the weak and original (strong) topology.
\end{lemma}
\begin{proof}
This is a direct consequence of Lemma~\ref{L:uX+A}, as $\Sch'$ is the countable union of weakly compact sets. Indeed, since the $\Sch$ is metrizable, there is a countable family $(\mathcal{U}_n)_{n\in\N}$ which forms a base of $0$-neighborhoods in $\Sch$. Then, according to the Alaoglu--Bourbaki theorem (\cite{Robertson1980}~Chapter~III~Theorem~6), the polars $\mathcal{K}_n:= \mathcal{U}_n^\circ\subset \Sch'$ are compact subsets of $\Sch'$ in the weak-* topology, such that  
\begin{eqnarray*}
    \Sch' & \quad=\quad & \bigcup_{n=1}^\infty \mathcal{K}_n.
\end{eqnarray*}
Then the statement follows from Lemma~\ref{L:uX+A} by choosing $X=Y=\Sch'$, and $u=C_\eta$ and $A=\{0\}$.
\end{proof}

In probability and measure theory, the notion of Polish spaces are rather important. These are such topological spaces, which are complete, separable, and metrizable. A generalization of such spaces are the Souslin spaces, which are images of Polish spaces by continuous maps. The significance of Souslin spaces in measure theory are given by the following collection of known theorems.

\begin{remark}\label{R:SouslinInj}
The notion of Borel sets over distribution theory function spaces is robust.
\begin{enumerate}[(i)]
    \item \label{R:SouslinInj:i} \cite{Bogachev1998}~Theorem~A.3.15~(iii) shows the following. Let $X$ and $Y$ be Hausdorff Souslin topological spaces, $C:\,X\to Y$ be an injective Borel measurable map. Then, the image of Borel sets in $X$ are Borel sets in $Y$. In particular, $C$ is a Borel isomorphism between $X$ and $C[X]\subset Y$.
    \item \label{R:SouslinInj:ii} It is well known, but \cite{Treves1970}~Proposition~A.9 and after explicitly shows that the function spaces $\EE$, $\Sch$, $\DD$ and their strong duals $\EE'$, $\Sch'$, $\DD'$ are Souslin spaces.
    \item \label{R:SouslinInj:iii} By construction, the continuous image of a Souslin space is Souslin. In particular, any topology on a Souslin space which is weaker than the original one is still Souslin.
    \item \label{R:SouslinInj:iv} A direct consequence of the above is that for any $\eta\in\Sch$ the subspace $C_{\eta}[\Sch']\subset\EE\cap\Sch'$ is Souslin with either of the subspace topology against $\EE$ or $\Sch'$.
    \item \label{R:SouslinInj:v} By taking the identity map as the Borel measurable injection in \enref{R:SouslinInj:i}, and also applying \enref{R:SouslinInj:iii}, it follows that in a Hausdorff Souslin space any weaker Hausdorff topology will generate the same Borel sets as the original topology. Consequently, on $\EE$, $\Sch$, $\DD$, $\EE'$, $\Sch'$, $\DD'$, the notion of Borel sets does not depend on which topology we consider between the strong (original) and the weak-* topologies. Moreover, in $C_{\eta}[\Sch']\subset\EE\cap\Sch'$ the notion of Borel sets does not depend on whether it is equipped with the subspace topology against any topology on $\EE$ or $\Sch'$ between their respective original or weak-* topologies. Using these arguments, one can also prove Lemma~\ref{L:CetaS} without referring to the metrizability of $\Sch$, which can come useful in future efforts of extending the results of this paper to the case of manifold spacetimes.
\end{enumerate}
\end{remark}

\section{Existence theorem of the UV limit measure}
\label{S:ExistenceMu}

In the theory of Wilsonian regularization, the properties of the convolution operators $C_{\eta}:\,\Sch'\to\Sch'$ (with some $\eta\in\Sch$) is central. Below some key results are recalled, on which our proofs will hinge.

\begin{remark}\label{R:ConvSurj}
Strong factorization properties of the convolution operator $\star:\,\Sch\times\Sch\to\Sch$.
\begin{enumerate}[(i)]
    \item \label{R:ConvSurj:i} In \cite{Voigt1984Factorization}~Theorem~3.2, or \cite{Garrett2004}, or \cite{Miyazaki1960}~Lemma~1 it is shown that the convolution operator on $\Sch$ is surjective, i.e.\ $\Sch\star\Sch=\Sch$. That is, for any $j\in\Sch$ there exists some $\eta,\ell\in\Sch$ such that $j=\eta\star\ell$ holds.
    \item \label{R:ConvSurj:ii} In addition, \cite{Garrett2004} and the proof of \cite{Miyazaki1960}~Lemma~1 shows that in the above factorization, one of the factors, say $\eta$, may be chosen to be $F(\eta)>0$. Quite obviously, the normalization of such $\eta$ may be chosen ensuring that $F(\eta)$ is unity at the origin. Due to the paracompactness of $\R^{N}$, for such $\eta$ one may construct a compactly supported smooth symmetric function $\alpha:\,\R^{N}\to[0,1]$ which is unity over a compact neighborhood of the origin containing the level set $\{F(\eta)\geq\frac{1}{2}\}$. Clearly, the function $\Phi:=(1-\alpha)\cdot 1+\alpha\cdot 1/F(\eta)$ is smooth, symmetric, everywhere positive, and is unity outside the support of $\alpha$, moreover $F^{-1}(\Phi\cdot F(\eta))\in\RegSp$, see again Definition~\ref{D:Reg}. Furthermore, for any $\ell\in\Sch$ one has $F^{-1}(1/\Phi\cdot F(\ell))\in\Sch$. Therefore, without loss of generality, the $\eta$ in \enref{R:ConvSurj:i} may be chosen to be $\eta\in\RegSp$.
    \item \label{R:ConvSurj:iii} In \cite{Voigt1984Factorization}~Theorem~3.2 it is shown that in $\Sch$ the so-called compact factorization property holds. Namely, for any compact set $\mathcal{J}\subset\Sch$ there exists some $\eta\in\Sch$ and some compact set $\mathcal{L}\subset\Sch$, such that $\mathcal{J}=\eta\star\mathcal{L}$ holds.
    \item \label{R:ConvSurj:iv} The straightforward combination of \enref{R:ConvSurj:ii} and \enref{R:ConvSurj:iii} evidently yields that the compact factorization property holds in a stronger form as well: one may choose $\eta$ such that $F(\eta)>0$, or one may even choose $\eta\in\RegSp$.
\end{enumerate}
\end{remark}

For our practical uses, we distill the following key lemmas as a consequence of the above known results recalled from the literature.

\begin{lemma}\label{L:Surj}
(factorization)
For every $j\in\Sch$ there exist $\eta\in\Sch$ and $\ell\in\Sch$, such that $j=\eta\star\ell$ holds. (This $\eta$ exists also such that $\eta\in\RegS$ or even $\eta\in\RegSp$.)
\end{lemma}
\begin{proof}
This is merely a summary restatement of Remark~\ref{R:ConvSurj}~\enref{R:ConvSurj:ii}.
\end{proof}

\begin{lemma}\label{L:Injectivity}
Let $\eta\in\Sch_{\C}$ with nowhere vanishing $F(\eta)$. Then, the convolution 
operator $C_{\eta}:\,\Sch_{\C}'\to\Sch_{\C}',\,t\mapsto C_{\eta}\,t:=\eta\star t$ is injective. In particular, then $C_{\eta}:\,\Sch_{\C}\to\Sch_{\C}$ is injective.
\end{lemma}
\begin{proof}
    Take some $\eta\in\Sch_{\C}$ and $t\in\Sch_{\C}'$. Then the Fourier transform of the convolution of $\eta$ and $t$ is $F(\eta\star t)=F(t)F(\eta)$ as tempered distribution, meaning that for all $\varphi\in\Sch_{\C}$ one has $(F(\eta\star t)\,|\,\varphi)=(F(t)\,|\,F(\eta)\cdot\varphi)$. In particular, for all $\varphi\in\DD_{\C}\subset\Sch_{\C}$ that identity holds. Consider now such an $\eta$ which has nowhere vanishing $F(\eta)$. Then for any $\varphi\in\DD_{\C}$ the function $\varphi/F(\eta)$ is everywhere defined, smooth and compactly supported, i.e. $\varphi/F(\eta)\in\DD_{\C}$. Consequently, for all $\varphi\in\DD_{\C}$ one has $(F(\eta\star t)\,|\,\varphi/F(\eta))=(F(t)\,|\,\varphi)$. If $t\in\Sch'$ is such that $\eta\star t=0$ holds, then by means of our observation, for all $\varphi\in\DD_{\C}$ the identity $(F(t)\,|\,\varphi)=(F(\eta\star t)\,|\,\varphi/F(\eta))=0$ needs to hold. That is, the tempered distribution $F(t)\in\Sch_{\C}'$ is zero on the dense subspace $\DD_{\C}\subset\Sch_{\C}$, implying that $F(t)=0$. By the injectivity of the Fourier transformation, $t=0$ follows.
\end{proof}

\begin{lemma}\label{L:Seq}
(sequential factorization)
Let $\seq{j}$ be a sequence in $\Sch$ that converges to some $j\in \Sch$. Then there exist a function $\eta\in\Sch$ and a sequence $\seq{\ell}$ in $\Sch$ convergent to some $\ell\in\Sch$, such that $j_n=\eta\star \ell_n$ holds for every $n\in\N$. (This $\eta$ exists also such that $\eta\in\RegS$ or even $\eta\in\RegSp$.)
\end{lemma}
\begin{proof}
    By assumption,  $\mathcal{J}:= \left\{j_n\,\vert\,n\in\N\right\}\cup\{j\}$ is a compact set in $\Sch$. Applying Remark~\ref{R:ConvSurj}~\enref{R:ConvSurj:iv} it follows that there exits some $\eta\in\RegSp$ and some compact set $\mathcal{L}\subset\Sch$, such that $\mathcal{J}=\eta\star\mathcal{L}$ holds. By this observation, for every $n\in\N$ there is some $\ell_{n}\in\mathcal{L}$ such that $j_{n}=\eta\star\ell_{n}$ holds. Since the sequence $\seq{\ell}$ runs in a compact set, therefore it is bounded and has accumulation points. We will show that all its accumulation points are the same, which implies that it is convergent. For let $\lambda_{1},\lambda_{2}\in\Sch$ be accumulation points. Then, there exists some index subsequence $(n_{k})_{k\in\N}$ such that $(\ell_{n_{k}})_{k\in\N}$ converges to $\lambda_{1}$, and there exists some index subsequence $(m_{l})_{l\in\N}$ such that $(\ell_{m_{l}})_{l\in\N}$ converges to $\lambda_{2}$. For these, one has that for all $k\in\N:$ $\eta\star\ell_{n_{k}}=j_{n_{k}}$, moreover for all $l\in\N:$ $\eta\star\ell_{m_{l}}=j_{m_{l}}$. By the continuity of convolution it follows then that $\eta\star\lambda_{1}=j$ and $\eta\star\lambda_{2}=j$, i.e.\ one has $\eta\star\lambda_{1}=\eta\star\lambda_{2}$. Since one had $F(\eta)>0$, the convolution by $\eta$ is injective by means of Lemma~\ref{L:Injectivity} and thus $\lambda_{1}=\lambda_{2}$ follows. Therefore, $\seq{\ell}$ converges to some $\ell:=\lambda_{1}=\lambda_{2}$.
\end{proof}

\begin{remark}\label{R:ApproxIdS}
The following results on approximate identities are recalled from literature.
\begin{enumerate}[(i)]
  \item \label{R:ApproxIdS:i} The convergence of a sequence of test functions $\alpha_{n}\in\DD$ ($n\in\N$) to the Dirac delta distribution $\delta\in\DD'$ in the $\DD'$ topology guarantees that for all test functions $\varphi\in\DD$ one has $\alpha_{n}\star\varphi\to\varphi$ pointwise as $n\to\infty$. It is not guaranteed, however, that the convergence also holds in the $\DD$ topology. On the other hand, it is well known that one can construct special Dirac delta approximating sequences $\alpha_{n}\in\DD$ ($n\in\N$) for which the pertinent convergence holds in $\DD$. These are called \defin{approximate identities}, see e.g.\ \cite{Rudin1991}~Theorem~6.32(a).
  \item \label{R:ApproxIdS:ii} It is common knowledge that the space $\Sch$ also admits approximate identity in the following sense. Sequences $\alpha_{n}\in\RegS\subset\Sch$ ($n\in\N$) exist such that for all $\varphi\in\Sch$ one has $\alpha_{n}\star\varphi\to\varphi$ in the $\Sch$ topology as $n\to\infty$. For let us take any $\alpha\in\RegS$ and define its compressed version $\alpha^{}_{\Lambda}(x):=\Lambda^{N}\,\alpha(\Lambda\,x)$ $\,$ ($x\in\R^{N}$) for all $1\leq\Lambda<\infty$. One can verify that for all $\varphi\in\Sch$ one has $\alpha^{}_{\Lambda}\star\varphi\to\varphi$ in the $\Sch$ topology as $\Lambda\to\infty$.  To confirm this, one should note that $F:\,\Sch_{\C}\to\Sch_{\C}$ is topological automorphism, and therefore it is enough to show that $(1-F(\alpha^{}_{\Lambda}))\cdot F(\varphi)\to 0$ in the $\Sch$ topology as $\Lambda\to\infty$. The canonical topology of $\Sch$ is characterized by the family of seminorms $\sup_{x\in\R^{N}}|Q(x)\,P(\partial_{x})\,\psi(x)|$ on a $\psi\in\Sch_{\C}$, with $P,Q$ being any polynomials. Clearly, the seminorms $\sup_{x\in\R^{N}}|P(\partial_{x})\,Q(x)\,\psi(x)|$ ($\psi\in\Sch_{\C}$) define the very same topology. Due to H\"older's inequality and the Sobolev inequality, this family of seminorms is equivalent to the family of seminorms $\Vert P(\partial)\,Q\,\psi\Vert_{L^{2}}$ ($\psi\in\Sch_{\C}$), with $P,Q$ being polynomials. Using $\left.F(\alpha^{}_{\Lambda})\right\vert_{\omega}=\left.F(\alpha)\right\vert_{\omega/\Lambda}$ ($\omega\in\R^{N}$), via Lebesgue's theorem of dominated convergence it is not hard to see that for all polynomials $P,Q$ and all $\alpha\in\RegS$, $\varphi\in\Sch$ one has $\Vert P(\partial)\,Q\cdot(1-F(\alpha^{}_{\Lambda}))\cdot F(\varphi)\Vert_{L^{2}}\to 0$ as $\Lambda\to\infty$, which completes the argument.
\end{enumerate}
\end{remark}

\begin{remark}
\label{R:BochnerMinlos}
Recall that given some sigma-additive non-negative valued finite measure $\mu$ on the Borel sets of $\Sch'$ it has a corresponding Fourier transform function $Z:\,\Sch\to\C$ defined by the formula
\begin{eqnarray}
 Z(j) & := & \int\limits_{\phi\in\Sch'} \e^{\I\,(\phi|j)} \,\de\mu(\phi) \qquad \mbox{ for all } j\in\Sch.
\end{eqnarray}
In the QFT context, this is sometimes called the \defin{partition function}, or in probability theory the \defin{characteristic function}. It has some notable properties, namely that it uniquely characterizes the measure $\mu$, it is continuous, and is a so-called positive definite function. The latter means that for all $m\in\N$, for all finite systems $j_{1},j_{2},\ldots,j_{m}\in\Sch$, the $m\times m$ matrix $\big(Z(j_k-j_l)\big)_{1\leq k,l\leq m}$ is positive semidefinite.  The Bochner--Minlos theorem (see e.g.\ \cite{Velhinho2017}~Corollary1, or \cite{Bogachev2007}~vol.2~Theorem7.13.9, or \cite{Bogachev2017}~Corollary5.11.11) says that the converse is also true. Namely, given a continuous positive definite function $Z:\,\Sch\to\C$ there is a unique sigma-additive non-negative valued finite measure $\mu$ on the Borel sets of $\Sch'$, such that the Fourier transform of $\mu$ equals to $Z$. Moreover, $\mu(\Sch')=Z(0)$ holds.
\end{remark}

\begin{remark}
\label{R:WRGZ}
Let $(\mu_{\eta})_{\eta\in\Reg}$ be a Wilsonian RG flow as in Definition~\ref{D:WRG}. Then, their corresponding family of Fourier transform functions $(Z_{\eta})_{\eta\in\Reg}$, with $Z_{\eta}:\,\Sch'\to\C$ for each $\eta\in\Reg$ being continuous and positive definite, obey the relation
\begin{eqnarray}\label{E:WRGZ}
\forall\, \eta,\eta',\eta''\in\Reg \mbox{ satisfying } C_{\eta''}=C_{\eta'}\,C_{\eta}: \qquad Z_{\eta''}=Z_{\eta}\circ C_{\eta'^{t}}.
\end{eqnarray}
Here, $\eta'^{t}$ stands for the reflected $\eta'$, i.e.\ for all $x\in\R^N$ one has $\eta'^{t}(x):=\eta'(-x)$. This equation is the simple consequence of the Wilsonian RG equation \eqref{E:WRG}, and the relation between measure pushforward and function composition, i.e.\ the fundamental formula for integration variable substitution, moreover the relation $(C_{\eta'}\,\phi\,|\,j)=(\phi\,|\,C_{\eta'^{t}}j)$ for all $\phi\in\Sch'$, $j\in\Sch$.
\end{remark}

In the following, we prove that the family of Fourier transforms $(Z_{\eta})_{\eta\in\Reg}$ of a Wilsonian RG flow $(\mu_{\eta})_{\eta\in\Reg}$ of sigma-additive non-negative valued finite measures has a corresponding continuous positive definite function $Z:\,\Sch\to\C$, such that $Z_{\eta}=Z\circ C_{\eta^{t}}$ holds for $\eta\in\Reg$.

\begin{theorem}\label{T:ZExists}
    Let $\Reg=\Sch$ or $\RegS$ or $\RegSp$, and $(Z_\eta)_{\eta\in \Reg}$ be a family of $\Sch\to\C$ continuous and positive definite functions, which satisfies \eqref{E:WRGZ}, i.e.\ has the property that
    \begin{eqnarray}\label{E:RGZ}
        Z_{\eta\star\eta'}(j)=Z_{\eta'}(\eta^{t}\star j)\qquad \mbox{for all}\quad \eta,\eta'\in\Reg \mbox{ and }j\in \Sch.
    \end{eqnarray}
    Then there exists a unique continuous positive definite function $Z:\,\Sch\to \C$ such that 
    \begin{eqnarray}\label{E:ZUV}
        Z_\eta(j)=Z(\eta^t \star j)\qquad\mbox{for all}\quad \eta\in\Reg \mbox{ and } j\in \Sch.
    \end{eqnarray}
\end{theorem}
\begin{proof}
    We are going to construct the function $Z:\,\Sch\to\C$ in a pointwise manner. For let $j\in\Sch$ be arbitrary but fixed, and choose $\eta\in\Reg$ and $\ell\in\Sch$ such that $j=\eta\star\ell$. Such factorization exists according to Lemma~\ref{L:Surj}. Define $Z(j)$ by letting
    \begin{eqnarray}\label{E:Z(j)}
        Z(j):= Z_{\eta^t}(\ell).
    \end{eqnarray}    
    First, we show below that the definition of $Z(j)$ does not depend on the choice of the factorization $j=\eta\star\ell$.

    Let us consider two factorizations $j=\eta_1\star\ell_1=\eta_2\star\ell_2$, where $\eta_{1},\eta_{2}\in\Reg$ and $\ell_{1},\ell_{2}\in\Sch$. Choosing any $\alpha\in\Reg$, by \eqref{E:RGZ} and by the commutativity of convolution
    \begin{eqnarray}
         Z^{}_{\eta_1^t}(\ell_1\star\alpha)=Z^{}_{\alpha^{t}\star\eta_1^t}(\ell_1)=Z^{}_{\alpha^t}(\eta_1\star\ell_1)=Z^{}_{\alpha^t}(j)\cr
         \Big.\qquad\qquad\qquad\qquad =Z^{}_{\alpha^t}(\eta_2\star\ell_2)=Z^{}_{\alpha^{t}\star\eta_2^t}(\ell_2)=Z^{}_{\eta_2^t}(\ell_2\star\alpha)
    \end{eqnarray}
    follows. Moreover, for any approximate identity $\alpha_{n}\in\Reg\subset\Sch$ ($n\in\N$) in $\Sch$, one has $\ell_1\star\alpha_{n}\to\ell_1$ and $\ell_2\star\alpha_{n}\to\ell_2$ in the canonical $\Sch$ topology as $n\to\infty$. Therefore, by continuity of the functions $Z^{}_{\eta_1^t}$ and $Z^{}_{\eta_2^t}$, one has $Z^{}_{\eta_1^t}(\ell_1)=Z^{}_{\eta_2^t}(\ell_2)$. This shows that $Z$ is well defined. 

    It is straightforward that $Z$ satisfies \eqref{E:ZUV}. For if $\eta\in\Reg$ and $\ell\in\Sch$, then by \eqref{E:Z(j)} we have
    $Z(\eta^t \star \ell)=Z^{}_{(\eta^t)^t}(\ell)=Z^{}_{\eta}(\ell)$.

    Next we show the continuity of $Z$, which is equivalent to its  sequentially continuity, since $\Sch$ is metrizable. Let $\seq{j}$ be a sequence in $\Sch$ such that $j_n\to j$ for some $j\in\Sch$. According to Lemma~\ref{L:Seq}, there exists a function $\eta\in\Reg$ and a sequence $\seq{\ell}$ in $\Sch$ such that $\ell_{n}\to\ell$ for some $\ell\in\Sch$, alongside $j_{n}=\eta\star\ell_{n}$ for all $n\in\N$. By the continuity of convolution, the identity $j=\eta\star\ell$ holds between the limits $j$ and $\ell$. Then, by \eqref{E:Z(j)} one has
    \begin{eqnarray}Z(j_n)=Z(\eta\star\ell_n)=Z^{}_{\eta^t}(\ell_n) \;\to\; Z^{}_{\eta^{t}}(\ell)=Z(\eta\star\ell)=Z(j) \quad (\mbox{as } n\to\infty),\cr
    \end{eqnarray}
    where we used the continuity of $Z^{}_{\eta^{t}}$.
    
    As a next step, we show that $Z^{}_{\eta^{}_{n}}(j)\to Z(j)$ as $n\to\infty$, whenever $\eta_{n}\in\Reg\subset\Sch$ ($n\in\N$) is an approximate identity in $\Sch$, and $j\in\Sch$. Indeed, for an approximate identity $\seq{\eta}$ one has $\eta^{t}_{n}\star j\overset{\Sch}{\to}j$ as $n\to\infty$, therefore 
    \begin{eqnarray}
        Z^{}_{\eta^{}_{n}}(j)= Z(\eta^{t}_{n}\star j)\to Z(j) \qquad (\mbox{as }n\to\infty)
    \end{eqnarray}
    because of the continuity of $Z$, just shown in the above paragraph.
    
    To conclude the proof, we show that the function $Z$ is positive definite. Consider a finite system $j_1,j_2,\ldots,j_m$ in $\Sch$ and fix an $\eta\in\Reg$. Let us introduce the $m\times m$ matrices
    \begin{eqnarray}
    \mathbf{Z}_{\eta}:= \big(Z_{\eta}(j_k-j_l)\big)_{1\leq k,l\leq m} \qquad \mbox{and}\qquad \mathbf{Z}:= \big(Z(j_k-j_l)\big)_{1\leq k,l\leq m}.
    \end{eqnarray}
    By assumption, $Z_{\eta}$ is a positive definite function, meaning that the matrix $\mathbf{Z}_\eta$ is positive semidefinite. Taking now an approximate identity $\eta_{n}\in\Reg\subset\Sch$ ($n\in\N$), by letting $n\to\infty$ we have $\mathbf{Z}^{}_{\eta^{}_{n}}\to \mathbf{Z}$  according to the observation in the previous paragraph. Since the set of positive semidefinite matrices is closed, the matrix $\mathbf{Z}$ is positive semidefinite too, which means that the function $Z$ itself is positive definite.
\end{proof}

According to the Bochner--Minlos theorem mentioned in Remark~\ref{R:BochnerMinlos}, and taking into account Remark~\ref{R:WRGZ} as well as the above Theorem~\ref{T:ZExists}, the following corollary can be stated.

\begin{corollary}\label{C:muExists}
Let $\Reg=\Sch$ or $\RegS$ or $\RegSp$, and $(\mu_{\eta})_{\eta\in\Reg}$ be a family of sigma-additive non-negative valued finite measures with the Wilsonian RG property as in Definition~\ref{D:WRG}. Then there exists a unique sigma-additive non-negative valued finite measure $\mu$ on the Borel sets of $\Sch'$, such that the factorization property
\begin{eqnarray}
\mu_{\eta}=(C_{\eta})_{*}\;\mu\qquad\mbox{for all}\quad \eta\in \Reg
\end{eqnarray}
holds. The pertinent $\mu$ will be called the \defin{UV limit of the RG flow of measures} $(\mu_{\eta})_{\eta\in\Reg}$.
\end{corollary}

\begin{remark}\label{R:weakLimit}
Assume that the conditions of Corollary~\ref{C:muExists} hold. Then, for any sequence $h_{n}:\,\Sch'\to\R$ ($n\in\N$) of uniformly bounded measurable functions, and any bounded continuous function $g:\,\Sch'\to\R$, by Lebesgue's theorem of dominated convergence,
\begin{eqnarray}\label{E:weakLimit1}
 \lim\limits_{n\to\infty} \int\limits_{\Sch'} \vert h_{n}\vert\cdot\vert g\circ C_{\eta_{n}}-g\vert \,\de\mu \;=\; 0
\end{eqnarray}
holds for all approximate identities $\eta_{n}\in\Reg\subset\Sch$ ($n\in\N$). Consequently,
\begin{eqnarray}\label{E:weakLimit2}
 \lim\limits_{n\to\infty} \int\limits_{\Sch'} h_{n}\cdot(g\circ C_{\eta_{n}}-g) \,\de\mu \;=\; 0
\end{eqnarray}
follows. This, in particular, implies that for all approximate identities, the measure sequence $\mu^{}_{\eta^{}_{n}}$ ($n\in\N$) converges to $\mu$ weakly. That is seen by choosing $h_{n}=1$ for all $n$ and by using the fundamental formula relating pushforward and integration variable substitution in \eqref{E:weakLimit2}.
\end{remark}

\begin{remark}\label{R:Refinement}
The following can be stated on generalizations of the above results.
\begin{enumerate}[(i)]
 \item\label{R:Refinement:i} On manifold spacetimes the following can be said. At the price of a more elaborate argument using compact sets in $\Sch$ and their accumulation points, the proof of Theorem~\ref{T:ZExists} can also be carried out by solely referring to Remark~\ref{R:ConvSurj}~\enref{R:ConvSurj:i} and \enref{R:ConvSurj:iii}, without referring to the commutativity of convolution operators, and without referring to Fourier spectra of the regulators. That is, the key proof can be carried out without using arguments very specific to a flat spacetime. If the analogy of the factorization theorems in Remark~\ref{R:ConvSurj}~\enref{R:ConvSurj:i} and \enref{R:ConvSurj:iii} could be proven on manifolds, meaning that if all test functions in $\DD$ were factorizable into a coarse-graining operator acting on some other test function from $\DD$, the analogy of Theorem~\ref{T:ZExists} and its consequences would automatically generalize to measures on $\DD'$ over manifolds. These are open questions worth to investigate.
 \item\label{R:Refinement:ii} Staying on flat spacetime, but allowing strictly bandlimited regulators $\Reg=\RegD$ only, the following can be said. Introduce the dense linear subspace $\SchD$ of $\Sch$ consisting of such elements whose Fourier transform reside in $\DD$. Assign a vector topology to $\SchD$ which is just the $\DD$ topology on the Fourier transforms. Then, clearly $\SchD\cong\DD$ (the identification map being the Fourier transformation), and thus $\SchD$ is barrelled nuclear and therefore Bochner--Minlos theorem applies. Its continuous dual space $\SchD'$ equipped with the strong dual topology will contain $\Sch'$ as a dense subspace, and $\SchD'\cong\DD'$. The analogy of Lemma~\ref{L:Surj} and Lemma~\ref{L:Seq} holds: by the construction of $\DD$, for all $j\in\SchD$ there exists some $\eta\in\RegD$ and $\ell\in\SchD$ such that $j=\eta\star\ell$, moreover, for all convergent sequences $j_{n}\in\SchD$ ($n\in\N$) there exists some $\eta\in\RegD$ and a convergent sequence $\ell_{n}\in\SchD$ ($n\in\N$) such that $j_{n}=\eta\star\ell_{n}$ ($\forall\,n\in\N$). Consequently, the scheme of the proof of Theorem~\ref{T:ZExists} can be copied when replacing $\Sch$ by $\SchD$ and $\RegS$ by $\RegD$: it follows that there exists a corresponding positive definite sequentially $\SchD$-continuous function $Z:\,\SchD\to\C$ in the UV limit. If all the regularized measures $\mu_\eta$ ($\eta\in\Reg$) had second moments (which is always the case in QFT applications), using the RG equation, the estimate \cite{Bogachev2017}~(5.5.3) and the quadratic bound $1-\cos t\leq t^{2}$, by \cite{Laszlo2024}~Lemma21 it follows that $Z$ is also topologically $\SchD$-continuous, and thus Bochner--Minlos theorem can be invoked. In summary, with $\Reg=\RegD$, the UV limit measure $\mu$ of a Wilsonian RG flow still exists, but as a measure on $\SchD'$, having second moment. If this second moment distribution is tempered (i.e., does not blow up in frequency space faster than all polynomials), then $Z$ is $\Sch$-continuous, i.e.\ $\mu$ will be supported on $\Sch'$.
 \item\label{R:Refinement:iii} On flat spacetime, if the regularized measure instances in a Wilsonian RG flow is described by a free Gaussian measure modified by an interaction potential, then also sharp bandlimited cutoffs can be defined, the cutoff operator being a measurable map with respect to the free Gaussian measure. Once well defined, the Wilsonian RG flows with sharp and smooth bandlimited cutoffs can be shown to be in one-to-one correspondence. Details of the proof of \enref{R:Refinement:ii} as well as the sharp cutoff construction will be spelled out in a different paper, focusing entirely on strictly bandlimited regulators, as their treatment need rather different and involved techniques. Certain further support properties of the UV limit measures with strictly bandlimited regulators will be still mentioned in Section~\ref{S:ExistenceV}.
 \end{enumerate}
\end{remark}

Since the paper deals with Euclidean QFT models, a note on reflection positivity property is in order. 
Denote by $E_{\varepsilon+} := \left\{ x \in \mathbb{R}^{N} \,\big\vert\, x_1 > \varepsilon \right\}$ the open half-spacetime with margin $\varepsilon\geq 0$. Reflection positivity of a probability measure $\mu$ on $\Sch'$ would mean that for its Fourier transform $Z$ the following property holds: for all $j_{1},\dots,j_{m}\in\DD$ with their support in $E_{0+}$, the matrix $\big(Z(j_{k}-\theta\,j_{l})\big)_{1\leq k,l\leq m}$ is positive semidefinite, where $\theta$ denotes the reflection of the $x_{1}$ spacetime coordinate (see e.g.\ \cite{Glimm1987}~Section~6.1~(i)). The first question is: assuming the reflection positivity of the UV limit $\mu$ of a Wilsonian RG flow as in Corollary~\ref{C:muExists}, what does it imply for the regularized measures $(\mu_\eta)_{\eta \in \Reg}$ in the flow. Let $\Reg = \Sch$ and take an $\eta\in\Reg$ such that it is symmetric, non-negative valued, compactly supported, with its support contained in the unit ball of $\R^{N}$, and having unit integral. Its compressed version $(\eta^{}_{\Lambda})_{1\leq\Lambda<\infty}$ in terms of Remark~\ref{R:ApproxIdS}~\enref{R:ApproxIdS:i} will define a $\DD\to\DD$ approximate identity, which will also be an $\Sch\to\Sch$ approximate identity (see e.g.\ \cite{Reed1981}~p.326). It is easily read off that $\mu_{\eta^{}_{\Lambda}}$ will be reflection positive with margin $1/\Lambda$: for all $j_{1},\dots,j_{m}\in\DD$ with their support in $E_{1/\Lambda\,+}$, the matrix $\big(Z_{\eta^{}_{\Lambda}}(j_{k}-\theta\,j_{l})\big)_{1\leq k,l\leq m}$ is positive semidefinite. This follows simply from the identity
\begin{equation*}
\mathrm{supp}(\eta_{\Lambda}^t \star j)
\subset
\mathrm{supp}(\eta_{\Lambda}^t) + \mathrm{supp}(j),
\end{equation*}
the symmetricity of $\eta^{}_{\Lambda}$, and \eqref{E:ZUV}. 
Conversely, if there is an approximate identity of the above kind, such that each $\mu_{\eta^{}_\Lambda}$ is reflection positive with margin $1/\Lambda$, then $\mu$ will be reflection positive, by the weak convergence of $\mu_{\eta^{}_{\Lambda}}$ to $\mu$.
Similarly, if in addition, the regulator $\eta$ was chosen to be $O(N)$ invariant, the Euclidean invariance is inherited between the subflow $(\mu_{\eta^{}_{\Lambda}})_{1\leq\Lambda<\infty}$ and the UV limit $\mu$, by analogous arguments.

\section{Existence of UV limit of relative interaction potential}
\label{S:ExistenceV}

In QFT applications, one often needs to address the problematics of comparing Wilsonian RG flows of two related models. For instance, the flow of an interacting model, with respect to a reference flow corresponding to the Gaussian measure subordinate to the free Klein--Gordon operator $(-\Delta+m^{2})$, as outlined in Section~\ref{S:Intro}. In particular, let $\mu_{\eta}$ and $\gamma_{\eta}$ be non-negative valued finite sigma-additive measures on $C_{\eta}[\Sch']$ ($\eta\in\Reg$), and assume that they form two Wilsonian RG flows in terms of Definition~\ref{D:WRG}. Assume that the evolution of these two flows ``cross'' each-other in the sense that for a certain $\eta\in\Reg$ there exists some relative field renormalization factor $z_{\eta}\in\R^{+}$ between them, such that the measure $\mu_{\eta}$ is absolutely continuous to the reference measure $\tilde{\gamma}_{\eta}$ at this particular $\eta$, where $\tilde{\gamma}_{\eta}:=(z_{\eta})_{*}\,\gamma_{\eta}$ denotes the reference measure $\gamma_{\eta}$ re-expressed on the rescaled fields by the factor $z_{\eta}$. The absolute continuity, spelled out explicitly, means that there exists some Borel measurable function $f_{\eta}:\,C_{\eta}[\Sch']\to\Rbar_{0}^{+}$, such that $\mu_{\eta}=f_{\eta}\cdot\tilde{\gamma}_{\eta}$ holds. (In QFT, often this relative density function is expressed through its negative logarithm $V_{\eta}$, i.e.\ via $f_{\eta}=\e^{-V_{\eta}}$, called relative interaction potential.) Spelling out the ansatz $\mu_{\eta}=f_{\eta}\cdot\tilde{\gamma}_{\eta}$ means that for each Borel set $A$ of $C_{\eta}[\Sch']$ one has 
\begin{equation*}
    \mu_{\eta}(A)=\int\limits_{\varphi\in A}f_{\eta}(\varphi)\,\de\tilde{\gamma}_{\eta}(\varphi)=\int\limits_{\phi\in \overset{-1}{(z_{\eta}\,C_{\eta})}(A)}f_{\eta}(z_{\eta}\,C_{\eta}\phi)\,\de\gamma(\phi)
\end{equation*}
where at the last equality we used the fact that the family of reference measures $(\gamma_{\eta})_{\eta\in\Reg}$ form a Wilsonian RG flow (with UV limit $\gamma$ guaranteed by Corollary~\ref{C:muExists}). The typical task is to try to characterize the evolution and UV limit of the flow $(\mu_{\eta})_{\eta\in\Reg}$ in terms of the reference flow  $(\gamma_{\eta})_{\eta\in\Reg}$.\footnote{As outlined at the end of Section~\ref{S:Intro}, also formalized in Definition~\ref{D:WRG}, without loss of generality, the running field renormalization factor can be merged into the flow of measures in a Wilsonian RG flow. That is, one can adapt a normalization convention in which it does not appear explicitly in the notation, when considering only a single flow. In this convention, it re-appears as a running relative field renormalization factor, when comparing two different flows against each-other.} The answer to this question on physics ground is not evident at all, since the whole concept of Wilsonian RG flow was invented in order to overcome the difficulty of directly specifying a measure $\mu:=f\cdot\gamma$ in the UV limit, where the density $f$ would be defined by $f:=\e^{-V}$, with $V$ being some pointwise functional of distributional fields, such as $V(\phi)=g\,\int\phi^{4}$ according to the tentative definition. The development of the concept of Wilsonian RG flows was motivated by the expectation that such UV limit interaction potential $V$ may not exist in general, but only regularized interaction potentials $V_{\eta}$ should exist on the space of UV regularized fields $C_{\eta}[\Sch']$ ($\eta\in\Reg$). Our results will show that actually, the UV limit relative interaction potential $V$ exists, in a generalized sense.

\begin{lemma}\label{L:fFlow}
Let $\Reg=\Sch$ or $\RegS$ or $\RegSp$, and let $\mu_{\eta}$ and $\gamma_{\eta}$ be Wilsonian RG flows of non-negative valued finite sigma-additive measures on $C_{\eta}[\Sch']$ ($\eta\in\Reg$), in terms of Definition~\ref{D:WRG}, with corresponding UV limit measures $\mu$ and $\gamma$ guaranteed by Corollary~\ref{C:muExists}. Assume that at some particular $\eta\in\RegSp$, there exists a field rescaling factor $z_{\eta}\in\R^{+}$ and a measurable function $f_{\eta}:\,C_{\eta}[\Sch']\to\Rbar_{0}^{+}$ such that $\mu_{\eta}=f_{\eta}\cdot\tilde{\gamma}_{\eta}$ holds, with $\tilde{\gamma}_{\eta}:=(z_{\eta})_{*}\,\gamma_{\eta}$. Then, one has
\begin{eqnarray}\label{E:fmu}
(f_{\eta}\circ C_{\eta})\cdot\big((z_{\eta})_{*}\gamma\big) & \;=\; &\mu.
\end{eqnarray}
That is, $\mu$ is absolutely continuous with respect to the measure $(z_{\eta})_{*}\gamma$ with corresponding density function $f_{\eta}\circ C_{\eta}$.
\end{lemma}
\begin{proof}
   Let $\eta\in\RegSp$ as above. By assumption, one had
   \begin{eqnarray}\label{E:Cmulong}
   (C_{\eta})_{*}\,\mu \;=\; \mu_{\eta} \;=\; f_{\eta}\cdot\tilde{\gamma}_{\eta} \;=\; f_{\eta}\cdot((z_{\eta})_{*}\,\gamma_{\eta}) \;=\; f_{\eta}\cdot((z_{\eta})_{*}\,(C_{\eta})_{*}\,\gamma) \cr
   \qquad\bigg. \;=\; f_{\eta}\cdot ((z_{\eta}\,C_{\eta})_{*}\,\gamma) \;=\; f_{\eta}\cdot ((C_{\eta}\,z_{\eta})_{*}\,\gamma) \;=\; f_{\eta}\cdot ((C_{\eta})_{*}\,(z_{\eta})_{*}\gamma)
   \end{eqnarray}
   where at the 4-th equality Corollary~\ref{C:muExists} was used, and at the 6-th equality the commutativity of $z_{\eta}$ and $C_{\eta}$ was taken into account. Since $C_{\eta}:\,\Sch'\to C_{\eta}[\Sch']$ is Borel measurable (as it is continuous), moreover $\Sch'$ and $C_{\eta}[\Sch']$ are Souslin spaces (Remark~\ref{R:SouslinInj}~\enref{R:SouslinInj:ii}~and~\enref{R:SouslinInj:iv}), and since $\eta$ was chosen such that $C_{\eta}$ is injective (Lemma~\ref{L:Injectivity}), by Remark~\ref{R:SouslinInj}~\enref{R:SouslinInj:i} it follows that $C_{\eta}$ is Borel isomorphism. That is, $C_{\eta}{}^{-1}$ is Borel measurable as well. Taking now any Borel measurable function $h:\,\Sch'\to\Rbar_{0}^{+}$ the composite function $g:=h\circ C_{\eta}{}^{-1}$ will be $C_{\eta}[\Sch']\to\Rbar_{0}^{+}$ Borel measurable. Evaluating the integral of this $g$ against the leftmost and rightmost side of \eqref{E:Cmulong}, and using the fundamental formula relating the pushforward and the integration variable substitution,
   \begin{eqnarray}
   \int\limits_{\phi\in\Sch'} h(\phi)\,\de\mu(\phi) & \;=\; & \int\limits_{\phi\in\Sch'} h(\phi) \, f_{\eta}(C_{\eta}\,\phi) \,\de((z_{\eta})_{*}\gamma)(\phi)
   \end{eqnarray}
   follows, which completes the proof.
\end{proof}

\begin{remark}\label{R:dimtrans}
Assuming the conditions of Lemma~\ref{L:fFlow}, one may realize that given $(\mu_{\eta})_{\eta\in\Reg}$, the reference flow $(\gamma_{\eta})_{\eta\in\Reg}$, the field renormalization factor $z_{\eta}$, and the density function $f_{\eta}$ may be transformed by some constant field rescaling $c\in\R^{+}$ according to the rule $\gamma'_{\eta}:=c_{*}\,\gamma_{\eta}$, $z'_{\eta}:=z_{\eta}/c$, and $f'_{\eta}:=f_{\eta}$, such that the conditions of Lemma~\ref{L:fFlow} will still hold with $(\gamma'_{\eta})_{\eta\in\Reg}$, $z'_{\eta}$, $f'_{\eta}$. Since the multiplication by $c$ and the operator $C_{\eta}$ commutes (for any $\eta\in\Reg$), one has $(z'_{\eta})_{*}\,\gamma'_{\eta}=(z_{\eta})_{*}\,\gamma_{\eta}$, and $\gamma'=c_{*}\,\gamma$. This simultaneous redefinition freedom of the field scale and the reference flow is often called dimensional transmutation in QFT. The rationale behind this naming is that at any fixed particular $\eta\in\RegSp$ where the assumption of Lemma~\ref{L:fFlow} holds, one may set $c:=z_{\eta}$, and do the above transformation. When such transformation is done, after the omission of $()'$ in the notation, at that particular $\eta$ the field renormalization factor $z_{\eta}$ is unity by convention. With that normalization convention, by Lemma~\ref{L:fFlow} one has the relation $(f_{\eta}\circ C_{\eta})\cdot\gamma=\mu$ at that particular $\eta$, which implies the absolute continuity of $\mu$ with respect to $\gamma$. In summary, whenever the conditions of Lemma~\ref{L:fFlow} are satisfied, $\mu$ is absolutely continuous to $c_{*}\,\gamma$ with some suitably chosen $c\in\R^{+}$, and the constant $c$ may be merged into $\gamma$ by field normalization convention, without loss of generality.
\end{remark}

\begin{lemma}\label{L:zConstant}
Assume that the conditions of Lemma~\ref{L:fFlow} hold at some particular $\eta\in\RegSp$. Then, it holds for all $\eta\in\Reg$, with $z_{\eta}=const$. Furthermore, in the normalization convention $z_{\eta}=1$ by means of Remark~\ref{R:dimtrans}, it follows that $\mu$ is absolutely continuous to $\gamma$. That is, it follows that for all $\eta\in\Reg$ there exists a Borel measurable function $f_{\eta}:\,C_{\eta}[\Sch']\to\Rbar_{0}^{+}$, and an $f:\,\Sch'\to\Rbar_{0}^{+}$, such that $\mu_{\eta}=f_{\eta}\cdot\gamma_{\eta}$ and $\mu=f\cdot\gamma$ holds. (In QFT, often $f_{\eta}$ and $f$ are characterized by their negative logarithm $V_{\eta}$ and $V$, respectively, i.e.\ $f_{\eta}=\e^{-V_{\eta}}$ and $f=\e^{-V}$, which are called interaction potentials.)
\end{lemma}
\begin{proof}
By Lemma~\ref{L:fFlow}, assuming the normalization convention of Remark~\ref{R:dimtrans}, it follows that the UV limit measure $\mu$ is absolutely continuous with respect to $\gamma$. That, by Radon--Nikodym theorem (\cite{Rudin1987}~Theorem~6.10), is equivalent to the statement that for all Borel sets $B$ of $\Sch'$ one has ${\Big(\;\gamma(B)=0 \;\Rightarrow\; \mu(B)=0\;\Big)}$. In order to prove the absolute continuity of $\mu_{\eta}$ with respect to $\gamma_{\eta}$ at arbitrary $\eta\in\Reg$, we show that for all Borel sets $B_{\eta}$ of $C_{\eta}[\Sch']$ one has ${\Big(\;\gamma_{\eta}(B_{\eta})=0 \;\Rightarrow\; \mu_{\eta}(B_{\eta})=0\;\Big)}$. For let $B_{\eta}$ be a Borel set of $C_{\eta}[\Sch']$ such that $\gamma_{\eta}(B_{\eta})=0$ holds. This, by means of Corollary~\ref{C:muExists} is equivalent to $\gamma(\overset{-1}{C_{\eta}}(B_{\eta}))=0$. Using the absolute continuity of $\mu$ with respect to $\gamma$, just derived above, by Radon--Nikodym theorem it follows that $\mu(\overset{-1}{C_{\eta}}(B_{\eta}))=0$. Applying Corollary~\ref{C:muExists}, it follows then that $\mu_{\eta}(B_{\eta})=0$ holds, which completes the proof.
\end{proof}

Putting together Lemma~\ref{L:fFlow}, Lemma~\ref{L:zConstant}, and \eqref{E:fmu} the following can be concluded.

\begin{corollary}\label{C:VExists}
Let $\Reg=\Sch$ or $\RegS$ or $\RegSp$, and let $\mu_{\eta}$ and $\gamma_{\eta}$ be Wilsonian RG flows of non-negative valued finite sigma-additive measures on $C_{\eta}[\Sch']$ ($\eta\in\Reg$), with corresponding UV limit measures $\mu$ and $\gamma$ guaranteed by Corollary~\ref{C:muExists}. If for some $\eta\in\RegSp$ there exists $z_{\eta}\in\R^{+}$ and a Borel measurable function $f_{\eta}:\,C_{\eta}[\Sch']\to\Rbar_{0}^{+}$ such that $\mu_{\eta}:=f_{\eta}\cdot\tilde{\gamma}_{\eta}$ holds with $\tilde{\gamma}_{\eta}:=(z_{\eta})_{*}\,\gamma_{\eta}$, then this holds for all $\eta\in\Reg$ with a suitable $z_{\eta}$ and $f_{\eta}$, moreover $z_{\eta}$ can be chosen to be constant as a function of $\eta$. In the normalization convention of Remark~\ref{R:dimtrans}, i.e.\ when $z_{\eta}=1$ by choice, there exists a Borel measurable function $f:\,\Sch'\to\Rbar_{0}^{+}$ such that $\mu=f\cdot\gamma$. (In the QFT notation, $f_{\eta}=\e^{-V_{\eta}}$ and $f=\e^{-V}$ is used, and the pertinent $f$ and $V$ are the UV limit density function and interaction potential, respectively.) Furthermore, the following equivalent relations hold:
\begin{eqnarray}\label{E:fVRG}
 \forall\,\eta\in\RegSp: \qquad (f_{\eta}\circ C_{\eta})\cdot\gamma & \;=\; & f\cdot\gamma, \cr
 \Big.\forall\,\eta\in\RegSp: \qquad\qquad V_{\eta}\circ C_{\eta} & \;=\; & V \qquad (\gamma\mbox{-almost everywhere}).
\end{eqnarray}
\end{corollary}

\begin{remark}\label{R:zNonRunning}
Note that in the above corollary the choice $z_{\eta}=const$  is the only possible one, whenever $\gamma_{\eta}$ are defined by $\gamma_{\eta}:=(C_{\eta})_{*}\,\gamma$ with $\gamma$ being a Gaussian measure on $\Sch'$ ($\eta\in\Reg$). That is implied by the rigidity property of Gaussian measures, see Remark~\ref{R:GaussSuppNonzero}~\enref{R:GaussSuppNonzero:vi} and \enref{R:GaussSuppNonzero:vii}.
\end{remark}

\begin{remark}\label{R:dmNonRunning}
A reference Gaussian subordinate to a Wilsonian RG flow cannot run.
\begin{enumerate}[(i)]
 \item\label{R:dmNonRunning:i} Let $\Reg=\Sch$ or $\RegS$ or $\RegSp$ and assume that at each $\eta\in\Reg$ for the studied Wilsonian RG flow the ansatz $\mu_{\eta}=\e^{-V_{\eta}}\cdot\gamma_{\eta}$ holds, where the reference measure $\gamma_{\eta}$ is the pushforward by $C_{\eta}$ of the canonical Gaussian measure on $\Sch'$, associated to the Klein--Gordon operator ${(-d^{2}\Delta+m^{2})}$, as outlined in Section~\ref{S:Intro} ($d,m\in\R^{+}$). One may allow in addition the parameters $d$ and $m$ to be possibly $\eta$-dependent. Then, due to Remark~\ref{R:GaussSuppNonzero}~\enref{R:GaussSuppNonzero:vi} and \enref{R:GaussSuppNonzero:viii}, $d_{\eta}=const$ and $m_{\eta}=const$ must hold. That is, modulo the smoothing by $\eta\in\Reg$, a reference Gaussian subordinate to a Wilsonian RG flow cannot run.
 \item\label{R:dmNonRunning:ii} Let $\Reg=\RegD$, i.e.\ assume that only strictly bandlimited regulators are allowed. As shown in Remark~\ref{R:Refinement}~\enref{R:Refinement:ii}, if $(\mu_{\eta})_{\eta\in\Reg}$ is a Wilsonian RG flow under this condition, its UV limit still exists, but as a measure $\mu$ on the space $\SchD'$. Assume that such flow satisfies the above ansatz in \enref{R:dmNonRunning:i}. Then, from Remark~\ref{R:GaussSuppNonzero}~\enref{R:GaussSuppNonzero:vi} and \enref{R:GaussSuppNonzero:viii} it still follows that $d_{\eta}=const$ and $m_{\eta}=const$ must hold. That is, the reference Gaussian subordinate to a Wilsonian RG flow, modulo smoothing, cannot run even when $\Reg=\RegD$. It is an open question whether without further assumptions $\mu$ has support on $\Sch'\subset\SchD'$ under the ansatz in \enref{R:dmNonRunning:i}, and if so, whether $\mu$ is absolutely continuous to the reference UV Gaussian $\gamma$. We intend to investigate these questions in a paper in preparation, focusing entirely on the bandlimited regulator case.
\end{enumerate}
\end{remark}

\begin{remark}\label{R:QFTdiff}
In the usual informal QFT RG theory approach, the differential form of \eqref{E:fVRG} is used. First, some regulator $\eta\in\RegSp$ is chosen. Then, its compressed version is defined via $\eta^{}_{\Lambda}(x):=\Lambda^{N}\,\eta(\Lambda\,x)$ $\,$ ($x\in\R^{N}$) for all $1\leq\Lambda<\infty$. Such family is approximate identity in $\Sch$, see again Remark~\ref{R:ApproxIdS}~\enref{R:ApproxIdS:ii}. Then, the informal ODE $\frac{\de}{\de\Lambda}\,\bigg(\exp(-V_{\eta^{}_{\Lambda}}\circ C_{\eta^{}_{\Lambda}})\cdot\gamma\bigg)=0$ is attempted to be solved for the flow of potentials $V_{\eta^{}_{\Lambda}}$, given some initial data $V_{\eta^{}_{\Lambda}}$ at $\Lambda=1$. This can be called the differential form of the Wilsonian RG equation for the interaction potential, and can be justified knowing \eqref{E:fVRG}. In the usual informal QFT RG theory approach, a global solution to this informal ODE is looked for, admitting a $\Lambda\to\infty$ limit. Note however, that it is not evident from first principles when such a solution would exist. Moreover, the described informal procedure implicitly uses a number of tacit assumptions, the justification of which is not easy to establish. Furthermore, even if such an ODE solution could be produced, it is merely a sufficient condition for a solution to \eqref{E:fVRG}.
\end{remark}

Despite of the difficulty to extrapolate the regularized interaction potential toward the UV infinity, certain properties of the UV limit can already be established from the regularized instances, as stated in the following.

\begin{theorem}\label{T:Vbounded}
Let $\Reg=\Sch$ or $\RegS$ or $\RegSp$, and let $\mu_{\eta}$ and $\gamma_{\eta}$ be Wilsonian RG flows of non-negative valued finite sigma-additive measures on $C_{\eta}[\Sch']$ ($\eta\in\Reg$), with corresponding UV limit measures $\mu$ and $\gamma$ guaranteed by Corollary~\ref{C:muExists}. Assume that at some $\eta\in\RegSp$ the conditions of Corollary~\ref{C:VExists} hold. Then, whenever the potential $V_{\eta}$ is bounded from below at that $\eta$, the UV limit potential $V$ is $\gamma$-essentially bounded from below, with the same bound. (Equivalently, if the density function $f_{\eta}$ is bounded from above at that $\eta$, then the UV limit density function $f$ is $\gamma$-essentially bounded from above, with the same bound.)
\end{theorem}
\begin{proof}
This can be directly read off from \eqref{E:fVRG}.
\end{proof}

\section{Case studies on example models}
\label{S:Examples}

As shown in Corollary~\ref{C:muExists}, a Wilsonian RG flow not terminating at some finite UV regularization strength always originates from a UV limit Feynman measure. That is, characterizing such flows is equivalent to characterizing their UV limit Feynman measures. The most common approach in the literature is to solve an informal ODE flow equation similar to Remark~\ref{R:QFTdiff}, and try to obtain a global solution admitting a UV limit. Besides such informal approach, a rigorous alternative is to make an educated guess directly on the UV limit Feynman measure. In the following we demonstrate this procedure on the $\varphi^{4}$ and similar models in spacetime dimensions 1 to 4, which need different techniques. Our guidelines will be Corollary~\ref{C:muExists}, Corollary~\ref{C:VExists} and Theorem~\ref{T:Vbounded}.

The most naive method is to construct the sought-after UV limit measure $\mu$ from a fixed reference (Gaussian) measure $\gamma$ on $\Sch'$ as described in Section~\ref{S:Intro}, and from an interaction potential $V:\,\Sch'\to\Rbar$, according to $\mu:=\e^{-V}\cdot\gamma$. Here $V$ is some kind of Borel measurable function, associated to the naive potential $\Vcal:\,\EE\cap\Sch'\to\Rbar$, usually defined by some local expression similar to that of $\Vcal(\varphi)=g\,\int_{\Omega}\varphi^{4}$, with $\Omega$ being some bounded open region of $\R^{N}$ (IR cutoff). The most simple association procedure of $V$ to $\Vcal$ is by smoothing (mollifying), i.e.\ taking the limit of $\mathcal{V} \circ C_{\eta^{}_{n}}$ in a suitable sense as $\seq{\eta}$ is an approximate identity in $\Sch$. 
To that end, let us introduce the set of mollifier sequences
\begin{eqnarray}
\Delta & := & \Big\{\seq{\eta} \mbox{ in } \Sch \;\Big\vert\; \seq{\eta} \;\mbox{ is approximate identity in } \Sch\Big\}
\end{eqnarray}
often used in the following.

\begin{example}[Extension by a single mollifier sequence]\label{Exa:singleeta}
Let $\gamma$ be a reference measure and $\Vcal:\,\EE\cap\Sch'\to\Rbar$ a Borel measurable functional that we wish to extend to $\Sch'$ and $\seq{\eta} \in \Delta$ a sequence.
Define the extension
\begin{equation}
V :\, \Sch' \to \Rbar, \qquad \phi \mapsto \liminf_{n \to \infty} \Vcal \left( \eta_n \star \phi \right) \, .
\end{equation}
This is clearly measurable (\cite{Rudin1987}~Theorem~1.14) and if one can control the above limit for a given potential, it is arguably one of the simplest extensions of $\Vcal$ to $\Sch'$.
\end{example}

As a specific mollifier sequence is singled out in the above construction, it is not canonical. A more canonical method is to simply use all of the mollifying sequences, and take the best of them.

\begin{example}[The greedy extension]
\label{Exa:greedy}
Let $\gamma$ be a reference measure and  $\Vcal:\,\EE\cap\Sch'\to\Rbar$ a Borel measurable functional on the smooth fields. Then define the measurable function $V$ via
\begin{eqnarray}\label{E:greedy}
  V & := & \mathop{(\gamma)\!\inf}\limits_{\seq{\eta}\in\Delta}\; \liminf\limits_{n\to\infty}\; \Vcal\circ C_{\eta_{n}}
\end{eqnarray}
which we call the \defin{greedy extension} of $\Vcal$ to $\Sch'$. Here $(\gamma)\!\inf$ is a shorthand for the infimum taken in the lattice of all $\gamma$-measurable functions, i.e.\ it is the ``functionwise'' infimum (see \cite{Hajlasz2002}~Lemma~2.6 and the preceding paragraphs, as well as in \cite{MeyerNieberg1991}~Lemma~2.6.1 and surrounding paragraphs). It is necessary to take the lattice (functionwise) infimum because the set $\Delta$ of mollifier sequences is uncountable, and therefore the pointwise infimum might lead to measurability issues.
\end{example}

Note the identity $\e^{-V}=\mathop{(\gamma)\!\sup}\limits_{(\eta_n)\in\Delta}\; \limsup\limits_{n\to\infty}\; \e^{-\Vcal\circ C_{\eta_{n}}}$ if $V$ is the greedy extension of $\Vcal$. Therefore, the above is the most economic measurable extension of potentials via smoothing, optimized such that $\e^{-V}$ has the largest possible overlap integral with $\gamma$ on all Borel sets. The following theorem summarizes its basic properties.

\begin{theorem}\label{T:greedy}
Let $\gamma$ be a reference measure, and let $\Vcal:\,\EE\cap\Sch'\to\Rbar$ be a measurable function bounded from below. Define $V$ to be its greedy extension according to \eqref{E:greedy}. Then $V$ is also a measurable function $\gamma$-essentially bounded from below, thus $\mu:=\e^{-V}\cdot\gamma$ defines a sigma-additive non-negative valued finite measure. The measure $\mu$ is not the zero measure if and only if
\begin{eqnarray}\label{E:gammasuplimsup}
\int\limits_{\phi\in\Sch'}\mathop{(\gamma)\!\sup}\limits_{(\eta_n)\in\Delta}\limsup\limits_{n\to\infty} \,\e^{-\Vcal(\eta_{n}\star\phi)}\,\de\gamma(\phi) \;>\; 0.
\end{eqnarray}
Moreover,
\begin{eqnarray}\label{E:suplimsup}
\sup\limits_{(\eta_n)\in\Delta}\limsup\limits_{n\to\infty} \int\limits_{\phi\in\Sch'}\e^{-\Vcal(\eta_{n}\star\phi)}\,\de\gamma(\phi) \;>\; 0
\end{eqnarray}
is sufficient for $\mu$ being not the zero measure.
\end{theorem}
\begin{proof}
    The first statement \eqref{E:gammasuplimsup} is evident from the definition of the greedy extension. The second statement \eqref{E:suplimsup} is seen as follows.

    Using the knowledge that $\Vcal$ is bounded from below, the family of functions $\e^{-\Vcal\circ C_{\eta}}$ ($\eta\in\Sch$) are bounded from above by some $K\geq 0$. That is, they have a $\gamma$-integrable majorant $K$. Taking now an approximate identity $(\eta_{n})\in\Delta$, and applying the Fatou lemma (\cite{Rudin1987}~Theorem~1.28) to the sequence of non-negative valued measurable functions $K-\e^{-\Vcal\circ C_{\eta_{n}}}$, then taking the negative of both sides of the Fatou inequality, and adding $\int_{\phi\in\Sch'}K\,\de\gamma(\phi)$, one infers
     \begin{eqnarray}\label{E:greedy1}
     \limsup\limits_{n\to\infty} \int\limits_{\phi\in\Sch'}\e^{-\Vcal(C_{\eta_{n}}\,\phi)}\,\de\gamma(\phi) \;\leq\;  \int\limits_{\phi\in\Sch'}\limsup\limits_{n\to\infty}\e^{-\Vcal(C_{\eta_{n}}\,\phi)}\,\de\gamma(\phi).
     \end{eqnarray}
     Since $\limsup\limits_{n\to\infty}\e^{-\Vcal\circ C_{\eta_{n}}}\leq\mathop{(\gamma)\!\sup}\limits_{(\eta_n)\in\Delta}\limsup\limits_{n\to\infty}\e^{-\Vcal\circ C_{\eta_{n}}}$ holds $\gamma$-almost everywhere, by applying the monotonity of integration to the right hand side of \eqref{E:greedy1}, one has
     \begin{eqnarray}
     \limsup\limits_{n\to\infty} \int\limits_{\phi\in\Sch'}\e^{-\Vcal(C_{\eta_{n}}\,\phi)}\,\de\gamma(\phi) \;\leq\;  \int\limits_{\phi\in\Sch'}\mathop{(\gamma)\!\sup}\limits_{(\eta_n)\in\Delta}\limsup\limits_{n\to\infty}\e^{-\Vcal(C_{\eta_{n}}\,\phi)}\,\de\gamma(\phi).
     \end{eqnarray}
     Taking now the $\sup\limits_{(\eta_{n})\in\Delta}$ of both sides of this inequality,
     \begin{eqnarray}\label{E:greedy2}
     \sup\limits_{(\eta_n)\in\Delta}\limsup\limits_{n\to\infty} \int\limits_{\phi\in\Sch'}\e^{-\Vcal(C_{\eta_{n}}\,\phi)}\,\de\gamma(\phi) \;\leq\;  \int\limits_{\phi\in\Sch'}\mathop{(\gamma)\!\sup}\limits_{(\eta_n)\in\Delta}\limsup\limits_{n\to\infty}\e^{-\Vcal(C_{\eta_{n}}\,\phi)}\,\de\gamma(\phi) \cr
     \qquad\qquad\qquad\qquad\qquad\qquad\qquad\qquad\qquad \;=\; \int\limits_{\phi\in\Sch'}\e^{-V(\phi)}\,\de\gamma(\phi)
     \end{eqnarray}
     is inferred. Thus, \eqref{E:suplimsup} is indeed a sufficient condition for the integral on the rightmost side of \eqref{E:greedy2} to be positive, which completes the proof.
\end{proof}

Considering the elementary properties of the supremum and limsup, \eqref{E:gammasuplimsup} and \eqref{E:suplimsup} may be simplified even further, as stated below.

\begin{corollary}\label{C:greedy}
Let $\gamma$ be a reference measure, and let $\Vcal:\,\EE\cap\Sch'\to\Rbar$ be a measurable function bounded from below. Let $V$ be its greedy extension according to \eqref{E:greedy}, and define the corresponding sigma-additive non-negative valued finite measure $\mu:=\e^{-V}\cdot\gamma$. Then, $\mu$ is not the zero measure if and only if there exists some approximate identity $\seq{\eta}$ in $\Sch$, such that
\begin{eqnarray}\label{E:Elimsup}
\int\limits_{\phi\in\Sch'}\limsup\limits_{n\to\infty}\e^{-\Vcal(\eta_{n}\star\phi)}\,\de\gamma(\phi) \;>\; 0
\end{eqnarray}
holds. Moreover,
\begin{eqnarray}\label{E:Elim}
\lim\limits_{n\to\infty} \int\limits_{\phi\in\Sch'}\e^{-\Vcal(\eta_{n}\star\phi)}\,\de\gamma(\phi) \;>\; 0
\end{eqnarray}
is sufficient for $\mu$ being not the zero measure.
\end{corollary}

The above necessary and sufficient condition \eqref{E:Elimsup} for the nonvanishing of $\e^{-V}\cdot\gamma$ is typically hard to evaluate, given some $\Vcal$. The sufficient condition \eqref{E:Elim}, however, is in principle evaluable for concrete QFT models, by mapping $\Sch'$, $\gamma$ and the operations defining $\Vcal$ to sequence spaces. It is still an unresolved interesting question, that along which additional properties on $\Vcal$ the condition \eqref{E:Elim} would be also necessary. If such properties could be uncovered, it would provide a rather powerful tool for checking the renormalizability of concrete QFT models. For potentials having finite UV limit of their expectation value, i.e.\ in case
\begin{eqnarray}\label{E:Vlim}
\lim\limits_{n\to\infty} \int\limits_{\phi\in\Sch'}\Vcal(\eta_{n}\star\phi)\,\de\gamma(\phi) \;<\; \infty
\end{eqnarray}
holds for some approximate identity $\seq{\eta}$ in $\Sch$, Jensen's inequality (\cite{Rudin1987}~Theorem~3.3) implies that the sufficient condition \eqref{E:Elim} is satisfied. Evaluating this, it follows that for the IR cutoff $\varphi^{4}$ model over a 0+1 dimensional spacetime, the greedy extension defines a valid interacting Feynman measure.\footnote{It is of course well known in the literature that this model has a Feynman measure over $\R^{1}$ spacetime. That is simply because a free Klein--Gordon Gaussian measure on fields over $\R^{1}$ is supported on the continuous functions, see Remark~\ref{R:GaussSuppNonzero}~\enref{R:GaussSuppNonzero:ii} in \ref{A:GaussSupp}.} The condition \eqref{E:Vlim} is also satisfied for IR cutoff potentials in arbitrary dimensional spacetimes whenever the potential density is bounded both from below and above. Such an example is the potential $\Vcal(\varphi)=g\,\int_{\Omega}\frac{1}{1+l^{4}\,\varphi^{4}}\,\varphi^{4}$ (``basin-shaped'' potential), or a sine--Gordon potential.\footnote{A very recent result on the triviality of Feynman measures subordinate to certain bounded potentials is proved in \cite{Dybalski2025}.} Therefore, with the greedy extension, these define their Feynman measures in a relatively straightforward way. The greedy extension preserves lower bound, and therefore is particularly attractive for extension of regularized potentials which are themselves bounded from below, as it respects the constraint by Theorem~\ref{T:Vbounded}.

Taking a step back, the use of mollifiers is actually just a way to approximate a given $\phi \in \Sch'$ by a sequence in $\EE \cap \Sch'$. If we simply consider every possible approximation at any point $\phi \in \Sch'$, this is a well-known method in convex analysis corresponding to taking the lower semicontinuous envelope (see e.g.\ \cite{Zalinescu2002}).

\begin{example}[The lower semicontinuous envelope]\label{Exa:lsc}
Let $\gamma$ be a reference measure and  $\Vcal:\,\EE\cap\Sch'\to\Rbar$ Borel measurable.
Set $W$ to be
\begin{equation}
W \left( \phi \right)
=
\cases{\Vcal \left( \phi \right) & if $\phi \in \EE \cap \Sch'$ \\
+\infty & otherwise,}
\end{equation}
i.e.\ to be the extension of $\Vcal$ by $+\infty$ wherever it was not already defined in $\Sch'$, and take its lower semicontinuous envelope
\begin{equation}
\underline{W} \left( \phi \right)
\;:=\;
\inf_{\phi_\alpha \to \phi} \liminf_\alpha W \left( \phi_\alpha \right)
\;=\;\inf_{\varphi_\alpha \to \phi} \liminf_\alpha \Vcal \left( \varphi_\alpha \right) \, ,
\end{equation}
where $(\phi_{\alpha})$ and $(\varphi_\alpha)$ are any net in $\Sch'$ and $\EE \cap \Sch'$, respectively, converging weakly to $\phi\in\mathcal{S}'$, and the $\liminf$ is understood as the infimum of all subsequential limits of the corresponding net.
The above construction corresponds to defining $\underline{W}$ by setting its epigraph to be equal to the closure of the epigraph of $W$, the closure taken with respect to the weak topology of $\Sch'$.

A straightforward consequence is that all sublevel sets $\overset{-1}{\underline{W}}( [-\infty, t] )$ for $t \in \R$ are weakly closed in $\Sch'$.
Hence, $\underline{W}$ is automatically Borel measurable.
\end{example}

A benefit of this construction is that it is always smaller than any (meaningful) limit of a sequence $\mathcal{V} \circ C_{\eta_n}$ ($n\in\N$), in particular it is smaller than the greedy extension. Consequently, $\e^{-\underline{W}}$ has a greater chance of not being zero $\gamma$-almost everywhere, in contrast to a corresponding limit of $\mathcal{V} \circ C_{\eta_n}$. If $\Vcal$ was bounded from below, then so is its lower semicontinuous envelope $\underline{W}$, and thus $\e^{-\underline{W}}\cdot\gamma$ defines a sigma-additive non-negative valued finite measure. Similar to the greedy extension, this product is nonzero whenever the overlap integral of $\e^{-\underline{W}}$ and $\gamma$ is nonzero.

It is apparent that the construction of UV limit potentials in general is very difficult.
In that spirit, below we collect some frequent pathology properties of a potential $\mathcal{V}$, such that even taking its lower semicontinuous envelope $\underline{W}$ as above does not help to define a nonzero measure $\e^{-\underline{W}}\cdot\gamma$.

\begin{lemma}\label{L:C+jX-Borel}
    Let $Y$ be a reflexive locally convex topological vector space. Let $Z\subset Y$ be a linear subspace equipped with a norm $\|\cdot\|$ such that the inclusion $j:\,(Z,\|\cdot\|)\to Y$ is continuous. Then, for any weakly closed subset $A\subset Y$, the subset $A+j^{**}[(Z,\|\cdot\|)'']$ is (weakly) Borel measurable in $Y$.
\end{lemma}
\begin{proof}
    The statement is straightforward from Lemma \ref{L:uX+A}. 
\end{proof}

In the following, we shall use the symbol $\overline{A}$ for the closure of a set $A$ in the locally convex topological vector space $Y$, with respect to the original topology. However,  if $A$ is a linear subspace (or merely a convex set) of $Y$, then weak closure coincides with the closure in the original topology (\cite{Rudin1991}~Theorem~3.12). In such case, the symbol $\overline A$ refers to either kind of closures.

\begin{theorem}\label{T:no-go}
Let $X$ be a vector space and $Y$ be a reflexive locally convex topological vector space. Let $\iota,\pi:\,X\to Y$ be linear operators, with $\iota$ being injective, such that 
\begin{eqnarray}\label{E:iota-pi}
(\iota-\pi)[X] & \quad\subset\quad & J:= \overline{\iota[\Ker\pi]}
\end{eqnarray}
holds. Let $\|\cdot\|$ be a norm on $\pi[X]$ such that the identity map $j: \, (\pi[X],\|\cdot\|) \to Y$ is continuous. 
Let $\Vcal:\,X\to\Rbar$ be some function, and define a corresponding set of values
\begin{equation}\label{E:L}
L:=
\Big\{ \liminf_{n \to \infty} \Vcal (\varphi_n) \,\Big\vert\, ( \varphi_n)_{n \in \mathbb{N}} \mbox{ in }  X \mbox{ with } \|\pi \varphi_n\|\to+\infty  \Big\}
\subset
\Rbar \, .
\end{equation}
Denote by
\begin{equation}
W :\, Y \to \Rbar, \qquad \phi \mapsto \cases{\Vcal(\iota^{-1}(\phi)) & if $\phi \in \iota[X]$ \\ +\infty & otherwise}
\end{equation}
the extension of $\Vcal\circ\iota^{-1}$ by $+\infty$ from $\iota[X]\subset Y$ to the full $Y$.

Then, $J + j^{**}[\pi[X]'']$ is a weak Borel set in $Y$, where $\pi[X]'':=(\pi[X], \|\cdot\|)''$ is the norm-bidual of $\pi[X]$. Furthermore, given some Borel measure $\gamma$ on the weak Borel sets of $Y$ obeying $\gamma({J + j^{**} [ \pi[X]''] }) = 0$, the lower semicontinuous envelope $\Wbar$ of $W$ takes values in $L \cup \{+ \infty \}$, $\gamma$-almost everywhere.
If in addition $\gamma( Y \setminus \overline{\iota[X]} ) = 0$, e.g.\ if $\iota[X]$ is (weakly) dense in $Y$, then $\Wbar$ takes values in $L$, $\gamma$-almost everywhere.
\end{theorem}
\begin{proof}
    First of all, the set $J + j^{**}[\pi[X]'']\subset Y$ is weakly Borel, according to Lemma~\ref{L:C+jX-Borel}. Let now $\phi\in Y$ be any vector. By definition we have 
\begin{eqnarray*}
    \Wbar(\phi)=\inf\Big\{\liminf_{\alpha} W(\phi_\alpha)\,\Big\vert\,(\phi_\alpha)_{\alpha\in I}\subset Y, \phi_\alpha\to \phi \quad\mbox{weakly in }Y \Big\},
\end{eqnarray*}
hence $\Wbar(\phi)=+\infty$ whenever $\phi\notin \overline{\iota[X]}$. On the other hand, when $\phi\in \overline{\iota[X]}$ we have 
\begin{eqnarray*}
    \Wbar(\phi)=\inf \Big\{
\liminf_{\alpha} \Vcal (\varphi_\alpha) \,\Big\vert\,
( \varphi_\alpha)_{\alpha \in I}\subset X,\,  \iota \varphi_\alpha\to \phi \quad\mbox{weakly in }Y 
\Big\}
\end{eqnarray*}
Hence, fix a net $(\varphi_\alpha)_{\alpha \in I}$ in $X$ such that $\iota \varphi_\alpha\to \phi$ weakly in $Y$. Suppose first that there exists a subnet $(\varphi'_\beta)_{\beta\in I'}$ such that $(\pi \varphi_{\beta}')_{\beta\in I'}$ is norm bounded in $(\pi[X],\|\cdot\|)$. Then, by the Banach--Alaoglu theorem (\cite{Rudin1991}~Theorem~3.15) we may suppose that the net $(\widehat{\pi \varphi'_\beta })_{\beta\in I'}$ has a weak-* limit $z\in\pi[X]''$, that is, 
\begin{equation*}
    \forall\, p\in (\pi[X],\|\cdot\|)'\;: \qquad \widehat{\pi \varphi'_\beta }(p):=p(\pi \varphi'_\beta)\to z(p).
\end{equation*}
Here $\widehat{(\cdot)}$ denotes the canonical injection of a vector into the continuous bidual space. 
Note that $j^{**}$ is equal to the weak-*-to-weak continuous extension of $j$ to the bidual space $\pi[X]''$, hence  $\pi(\varphi_\beta')=j\pi(\varphi_\beta')=j^{**}\widehat{\pi(\varphi_\beta')}\to j^{**}z$ weakly in $Y$. Since $\iota \varphi_\beta'\to \phi$ weakly in $Y$, it follows that $(\iota-\pi)\varphi_\beta'\to \phi-j^{**}z$ weakly in Y, thus
\begin{equation*}
    \phi-j^{**}z\in \overline{(\iota-\pi)[X]}\subset J 
\end{equation*}
 according to \eqref{E:iota-pi}. In other words, $\phi \in J + j^{**}[\pi[X]'']$. As a consequence we obtained that if $(\varphi_\alpha)_{\alpha\in I}$ is any net in $X$ such that $\iota \varphi_\alpha\to \phi$ weakly in $Y$, where  $\phi\in Y\setminus(J + j^{**}[\pi[X]''])$, then necessarily $\|\pi \varphi_\alpha\|\to+\infty$.

 Let now $\phi\in  Y\setminus(J + j^{**}[\pi[X]''])$ be any vector. It is then easy to see that $\underline {W}(\phi)=\inf A_{\phi}$ where 
 \begin{equation*}
     A_{\phi}:=\left\{\lim_\alpha \Vcal(\varphi_\alpha) \,\Big\vert\, \iota \varphi_\alpha\to \phi \mbox{ weakly in }Y\mbox{ and } (\Vcal(\varphi_\alpha))_{\alpha\in I} \mbox{ converges in }\Rbar \right\}.
 \end{equation*}
 We show that $A_{\phi}\subset L$. For let $a\in A_{\phi}$ be finite (the case when $a=\pm\infty$ is proved in an analogous way) and consider a net $(\varphi_\alpha)_{\alpha\in I}$ such that $\iota \varphi_\alpha\to \phi$ weakly in $Y$ and $\Vcal(\varphi_\alpha)\to a$. By the first half of the proof, we have $\|\pi \varphi_\alpha\|\to+\infty$. Consequently, for any fixed $n$ the index set $I_n:= \{\alpha\in I \,\vert\, \|\pi \varphi_\alpha\|>n\}$ is co-final in $I$, that is, $(\varphi_\alpha)_{\alpha\in I_n}$ is a subnet of $(\varphi_\alpha)_{\alpha\in I}$. Therefore, for every $n\in \mathbb N$ there exists $\beta_n\in I_n$ such that $a-\frac1n<\Vcal(\varphi_{\beta_n})<a+\frac1n$. This means that $\xi_n:= \varphi_{\beta_n}$ $(n\in \mathbb N)$ is a sequence such that $\|\pi \xi_n\|\to+\infty$ and $\Vcal(\xi_n)\to a$, thus $a\in L$, as it is claimed. Finally, since $L$ is obviously closed, it follows that   
 \begin{equation*}
     \Wbar(\phi)=\inf A_{\phi}\in \overline{A_{\phi}}\subset L,
 \end{equation*}
 which concludes the proof.
\end{proof}

As a corollary one can now see that the attempt to extend the $\varphi^4$ potential to $\Sch'$ directly, is bound to fail over $\R^{N}$ spacetimes with $N>1$.

\begin{corollary}\label{C:no-go}
Let $\iota:\,\Sch\to\Sch'$ denote the canonical embedding of Schwartz functions into the space of tempered distributions. Let $\Omega \subset \R^{N}$ be a bounded open set. 
Consider the function
\begin{equation}
\Vcal :\, \Sch \to \R, \qquad \varphi \mapsto \int\limits_\Omega |\varphi|^4.
\end{equation}
Denote by
\begin{equation}
W :\, \Sch' \to \Rbar, \qquad \phi \mapsto
\cases{\Vcal(\varphi) & when $\phi = \iota \varphi$ for some $\varphi \in \Sch$ \\
+\infty & otherwise}
\end{equation}
the extension of $\Vcal\circ\iota^{-1}$ from $\iota[\Sch]$ to $\Sch'$ by $+\infty$. 
Let $\gamma$ be the Gaussian measure on $\Sch'=\Sch'(\R^{N},\C)$ subordinate to the free Klein--Gordon operator $(-\Delta+m^{2})$ with $m > 0$. 

Then, whenever $N>1$, the lower semicontinuous envelope $\Wbar$ of $W$ in the weak topology is equal to $+\infty$, $\gamma$-almost everywhere. Consequently, any measurable function $V:\,\Sch'\to\Rbar$ obeying $V\geq\Wbar$, in particular the greedy extension of $\Vcal$, is equal to $+\infty$, $\gamma$-almost everywhere.
\end{corollary}
\begin{proof}
Let $\pi:\,\Sch \to \Sch'$ the mapping defined by $(\pi(\varphi)\,\vert\,\psi):=\int_{\Omega}\varphi\cdot\psi$ for all $\varphi,\psi\in\Sch$. That is, let $\pi$ be the IR cutoff version of $\iota$. Then, for all $\varphi\in\Sch$ one has
\begin{equation}
(\iota \varphi - \pi \varphi \,\vert\, \psi)
=
\int\limits_{\R^{N} \setminus \Omega} \varphi \cdot \psi \qquad(\forall\,\psi\in\Sch).
\end{equation}
Consequently, for our $\iota$ and $\pi$ the condition \eqref{E:iota-pi} is satisfied, and the notation $J:=\overline{\iota[\Ker\pi]}$ may be introduced. Choose the norm $\|\cdot\|$ on $\pi[\Sch]$ induced by the $L^4(\R^{N},\C)$ norm on $\Sch$, and with a slight abuse of notation denote by $L^{4}(\Omega)$ the completion of $\pi[\Sch]$ in that norm. Because Jensen's inequality, one has $L^{4}(\Omega)\subset L^{1}(\Omega)$, moreover every distribution from $J$ has support on $\R^{N}\setminus \Omega$. Therefore, every distribution in $J + L^4(\Omega)$ is also in $J+L^{1}(\Omega)$, i.e.\ is equal to a locally integrable function on $\Omega$. Consequently, by means of Remark~\ref{R:GaussSuppNonzero}~\enref{R:GaussSuppNonzero:iii}, $\gamma(J + L^4(\Omega)) = 0$ holds whenever $N>1$. Applying now Theorem~\ref{T:no-go}, it follows that $\Wbar$ takes its values in the set $L$ defined according to \eqref{E:L}, $\gamma$-almost everywhere. Finally, it is easy to check that $L=\{+\infty\}$ with our potential $\Vcal$ and norm $\|\cdot\|$, which completes the proof. This latter property is termed as the coercivity of $\Vcal$ with respect to the norm $\|\cdot\|$.
\end{proof}

Observe, that the above no-go theorem is avoided by the aforementioned basin-shaped potential, as well as the competing Higgs potentials of the form $\Vcal(\varphi,\psi)=g\,\int_{\Omega}(\varphi^{2}-\psi^2)^{2}$, since their corresponding characterizing sets $L$ from \eqref{E:L} will contain finite values (in particular, the zero). That is, the pertinent potentials are not coercive with respect to the $L^{4}$ norm. The IR cutoff $\varphi^{4}$ potential over 1 dimensional spacetime is of course coercive with respect to $L^{4}$, but the no-go theorem is avoided because $L^{4}(\Omega)\subset\Sch'$ is not zero $\gamma$-measure subspace in that setting.

\begin{example}[Extension by normal ordering]\label{Exa:no}
Let $\gamma$ be a fixed reference measure on $\Sch'$. If $n$ probability variables (real valued Borel measurable functions) $\Sch'\to\R$ are given, then for each $k=0,\dots,n$ their $k$-th degree normal ordered (also called Wick ordered) polynomials can be defined as certain linear combinations of the $0,\dots,k$-th degree ordinary polynomials (see also \cite{Frohlich2024}~Appendix~A.1). The normal ordered polynomials are distinguished by the fact that their $\gamma$-expectation value is zero whenever $k>0$. Starting out from the IR cutoff $\varphi^{4}$ potential $\Vcal(\varphi)=g\,\int_{\Omega}|\varphi|^{4}$ ($\varphi\in\EE\cap\Sch'$), its normal (Wick) ordered version at regulator $\eta\in\Sch$ is seen to be
\begin{eqnarray*}
 V_{n.o.,\eta}(\phi) & \;=\; & g\,\int\limits_{\Omega}\left( |\eta\star\phi|^{4} - 6\,\mathbb{E}_{\gamma}[|\eta\star\phi|^{2}]\,|\eta\star\phi|^{2} + 3\,\mathbb{E}_{\gamma}[|\eta\star\phi|^{2}]^{2} \right) \cr
  & & \qquad\qquad\qquad\qquad\qquad\qquad\qquad\qquad\qquad\qquad (\phi\in\Sch'),
\end{eqnarray*}
$\mathbb{E}_{\gamma}[\cdot]$ denoting the $\gamma$-expectation value. Over $\R^{2}$ spacetimes, it is known that $V_{n.o.,\eta}$ converges to some $V_{n.o.}$ in the $L^{2}_{\gamma}$ norm as $\eta\to\delta$ (as approximate identity), and that the corresponding $\e^{-V_{n.o.}}$ is $\gamma$-integrable (see a review in \cite{Frohlich2024}). Therefore, this $\e^{-V_{n.o.}}\cdot\gamma$ defines a Feynman measure for the IR cutoff $\varphi^{4}$ model over 2 dimensional spacetime. One should note, however, that the family of regularized measures $\e^{-V_{n.o.,\eta}}\cdot\gamma$, used in the above constructive limiting procedure, itself is not a Wilsonian RG flow: they are not marginals to one another with respect to some coarse-graining and field rescaling (as they violate Theorem~\ref{T:Vbounded}). Of course, after the limiting measure $\e^{-V_{n.o.}}\cdot\gamma$ is constructed by any procedure, it then induces a Wilsonian RG flow, and it is an interesting open question what relation it admits with respect to the family of measures $\e^{-V_{n.o.,\eta}}\cdot\gamma$ ($\eta\in\Sch$) used in the constructive limiting procedure.
\end{example}

It is also an interesting question whether e.g.\ the above well-known UV limit potential $V_{n.o.}$ of $\varphi^4$ over $\R^{2}$ is $\gamma$-almost everywhere equal to the lower semicontinuous envelope of itself. It would be also of great interest to deduce the marginal measures (Wilsonian RG flow) of the pertinent construction.

\begin{example}[Extension by running counterterms]\label{Exa:runningc}
For $\varphi^{4}$ models on $\R^{3}$, a slightly more general construction was demonstrated \cite{Barashkov2020,Barashkov2021}. Instead of the potential corrected with the normal ordering terms, general regulator dependent compensating terms were considered. At given regularization $\eta\in\Sch$ the regularized potential
\begin{eqnarray*}
 V_{a(\eta),b(\eta),\eta}(\phi) & \;=\; & \int\limits_{\Omega}\left(  g\, |\eta\star\phi|^{4} - a(\eta)\,|\eta\star\phi|^{2} + b(\eta) \right) \qquad\qquad (\phi\in\Sch')
\end{eqnarray*}
was taken, with the real valued continuous nonlinear functionals $a,b:\,\Sch\to\R$ of the regulators being not a priori fixed. These are sometimes called running couplings and their corresponding terms as counterterms. Clearly, normal ordering is a special case of this ansatz. The unknown functionals $a,b$ are fixed by additional conditions on certain observables, sometimes referred to as renormalization conditions. In \cite{Barashkov2020,Barashkov2021} it was shown that the $\varphi^{4}$ potential of the above kind over $\R^{3}$ spacetime admits a weak limit measure $\mu$ of the family of regularized measures $\e^{-V_{a(\eta),b(\eta),\eta}}\cdot\gamma$ as $\eta\to\delta$ (as approximate identity), whenever the functionals $a,b$ are suitably chosen. Although, the explicit form of the suitable functionals $a,b$ for this model are not yet known, it was shown that the limiting measure $\mu$ cannot be absolutely continuous with respect to the reference measure $\gamma$ one started from. Therefore, UV limit of the potential family $V_{a(\eta),b(\eta),\eta}$ ($\eta\in\Sch$) does not exist. Similarly to Example~\ref{Exa:no}, the pertinent family of measures used in the limiting procedure is not a Wilsonian RG flow itself (as they violate Corollary~\ref{C:muExists}), but given the limiting measure $\mu$ it induces a Wilsonian RG flow.
\end{example}

We note in passing that a very famous recent construction similar to the above was worked out in \cite{Aizenman2021} for $\varphi^{4}$ potential over $\R^{4}$, but with a lattice spacetime regularization. This case is not discussed in the present paper, as we concentrate on continuum spacetimes.

\begin{example}[Extension by running counterterms and reference measure]\label{Exa:runninggamma}
One could also consider an even more general limiting procedure such that both the regularized potential as well as the reference measure depends on the regularization $\eta\in\Sch$. This scenario can happen, for instance, if the reference measure $\gamma_{d(\eta),m(\eta)}$ is a Gaussian measure subordinate to the Klein--Gordon operator $(-d^{2}(\eta)\,\Delta+m^{2}(\eta))$, and where the $d$ parameter as well as the mass parameter $m$ is taken to be a not yet specified $\eta$-dependent coefficient, similar to that of the running couplings in the potential of Example~\ref{Exa:runningc}. The regularized measure $\e^{-V_{\eta}}\cdot\gamma_{d(\eta),m(\eta)}$ then may have some weak limit measure $\mu$ as $\eta\to\delta$ (as approximate identity), whenever the running couplings in $V_{\eta}$, $d(\eta)$ and the running mass $m(\eta)$ are suitably chosen. For this case, we are not aware of a mathematically worked out example in the literature. Again, such limiting procedure itself is non-Wilsonian, but the limit induces a Wilsonian RG flow.
\end{example}

In conclusion, based on Corollary~\ref{C:muExists}, Corollary~\ref{C:VExists} and Theorem~\ref{T:Vbounded}, one can state that if particularly the Wilsonian RG prescription is important for the constructive description of interactions appearing in Nature, and the interacting measure is to be described by a potential with respect to a Gaussian reference flow induced by a Klein--Gordon operator, then in 4 dimensional spacetimes certain interactions, such as the $\varphi^{4}$ are disfavored, whereas other interactions such as the basin-shaped potential, or the $(\varphi^{2}-\psi^{2})^{2}$ type competing Higgs model are favored.

\section{Concluding remarks}
\label{S:Conclusion}

In a Euclidean signature interacting QFT model, the vacuum state is described by the Wilsonian renormalization group (RG) flow of regularized Feynman measures. These are Feynman measures on the spaces of ultraviolet (UV) regularized fields, linked by a consistency condition called the Wilsonian renormalization group equation. The pertinent equation asserts that proceeding in the flow from the UV toward infrared (IR) means subsequent application of coarse-graining operators to fields, and that the corresponding regularized Feynman measure instances in the flow are marginals of each-other with respect to the intermediary coarse-graining (Definition~\ref{D:WRG}). Important QFT models are those which admit Wilsonian RG flows meaningful at any UV regularization strengths.

In this paper it was shown that the Wilsonian RG flow of Feynman measures extending to any UV regularization strengths admit a factorization property. Namely, there exists an ultimate (UV limit) Feynman measure on the unregularized distributional fields, such that the regularized Feynman measure instances in the flow are obtained from that ultimate measure, via marginals by coarse-graining (Corollary~\ref{C:muExists}).

In addition to proving the existence of the UV limit Feynman measure, it was shown that if two Wilsonian RG flows of Feynman measures admit a relative interaction potential with respect to one-another at a suitable coarse-graining, then the entirety of the two flows admit a flow of relative interaction potentials, moreover, there exists a UV limit relative interaction potential between the UV limit Feynman measures as well (Corollary~\ref{C:VExists}). Whenever the regularized relative interaction potential between two Wilsonian RG flows is bounded from below at a suitable coarse-graining, then the UV limit relative interaction potential is also bounded from below, with the same bound (Theorem~\ref{T:Vbounded}).

It was also shown that whenever a Wilsonian RG flow is described by a kinetic Gaussian measure (with possibly running parameters as a function of coarse-grainings) modified by a running interaction potential, then the parameters of this reference kinetic Gaussian measure cannot run. This behavior persists regardless of the allowed frequency cutoff tail shape of the coarse-graining operators (Remark~\ref{R:dmNonRunning}). In particular, for a Wilsonian RG flow, the field renormalization factor cannot run if the flow was defined by a kinetic Gaussian measure modified by some running potential.

Applications of the above theorems to concrete interacting QFT models were provided, over continuum flat Euclidean signature spacetimes, in dimensions 1 to 4 (Section~\ref{S:Examples}). In particular, we prove that in arbitrary dimensional spacetimes, an interacting QFT model defined by a free Klein--Gordon kinetic term + an interaction potential which is bounded both from below and above, is nonperturbatively renormalizable (Theorem~\ref{T:greedy}, Corollary~\ref{C:greedy}). Such models are defined e.g.\ by the ``basin-shaped'' potential $V(\varphi)=g\,\int\frac{1}{1+l^{4}\,\varphi^{4}}\,\varphi^{4}$, or by a sine--Gordon model. (On the triviality question of some of these models, see the very recent results of \cite{Dybalski2025}.) The nonperturbative renormalizability of further relevant models, such as the competing Higgs model defined by the potential $V(\varphi,\psi)=g\,\int(\varphi^{2}-\psi^{2})^2$ also could not be decided, but there are indications for its renormalizability (it avoids the no-go Theorem~\ref{T:no-go}, which obstructs high dimensional $\varphi^{4}$ models). The results proved in the paper were also discussed in the light of known examples, in particular the $\varphi^{4}$ models over 1 to 4 dimensional spacetimes \cite{Frohlich2024,Barashkov2020,Barashkov2021,Aizenman2021}.

\section*{Acknowledgements}
The authors would like to thank to J\'anos Balog, Gergely Fej\H{o}s, Antal Jakov\'ac, B\'alint T\'oth, Jan Derezi\'nski and Markus Fr\"ob for discussions and valuable insights.

This work was supported in part by the Hungarian Scientific Research fund, NKFI K-138136-138152, K-142423, ADVANCED-150038, ADVANCED-150059. Z.~Tarcsay was supported by  Project no.\  TKP2021-NVA-09, which has been implemented with the support provided by the Ministry of Culture and Innovation of Hungary from the National Research, Development and Innovation Fund, financed under the TKP2021-NVA funding scheme.

\appendix

\section{On the supports of Gaussian measures}
\label{A:GaussSupp}

For a sigma-additive non-negative valued measure $\gamma$ on the Borel sets of some topological space there is the well known notion of topological support, denoted by $\supp(\gamma)$, see e.g.\ \cite{Bogachev1998}~Definition~A.3.14. It is defined as the intersection of those closed sets whose complements are $\gamma$-measure zero. By construction, the topological support of a measure is a closed set, and is a generalization of the support notion of the continuous functions on topological spaces. Not surprisingly, typical Gaussian measures turning up in QFT are so-called non-degenerate, meaning that their topological support is the full space $\Sch'$. Since, however, in QFT the considered density functions $\e^{-V}$ to be integrated against the reference Gaussian measure $\gamma$ are not continuous, only measurable, the topological support of the reference measure $\gamma$ is not very informative on the non-vanishing of the overlap integral of $\e^{-V}$ and $\gamma$.

Due to the above complication, another notion of support is also used in the literature. A Borel set $B$ is said to support $\gamma$ if its complement has zero $\gamma$-measure. Clearly, if one can find some Borel set $B$ such that its $\gamma$-measure is nonzero, and the density function $\e^{-V}$ is nowhere zero on $B$, then their product $\e^{-V}\cdot\gamma$ is not the zero measure, i.e.\ it defines a proper probability measure modulo normalization, meaning that $V$ describes a renormalizable QFT model.

\begin{remark}
\label{R:GaussSuppNonzero}
In the following we recall some theorems on Borel sets from $\Sch'(\R^{N})$ which are known to support a Gaussian measure $\gamma$ with Fourier transform $j\mapsto\e^{-\frac{1}{2}(j\,\vert\,\KK\star j)}$ here $\KK$ denoting the maximally symmetric fundamental solution of a Euclidean Klein--Gordon operator $(-\Delta+m^{2})$ on $\R^{N}$. Also, theorems are recalled on those Borel sets which are known to have zero $\gamma$-measure against such Gaussian. See an excellent review in \cite{Velhinho2017,Velhinho2015,Bogachev1998}.
\begin{enumerate}[(i)]
 \item \label{R:GaussSuppNonzero:i} The space of Schwartz distributions $\Sch'$ is known to be a countable union of nested Hilbert spaces defined by the completion of $\Sch$ with respect to the scalar product $\left<\cdot,(I+\mathcal{N})^{k}\cdot\right>_{L^{2}}$ ($k\in\mathbb{Z}$), with $\mathcal{N}$ being the so-called harmonic oscillator number operator. For far enough indices, the canonical inclusion maps for far enough Hilbert spaces in this inclusion chain is Hilbert--Schmidt. As a consequence of this, the pertinent Gaussian measure $\gamma$ is supported within such Hilbert subspace within $\Sch'$, for low enough index $k<0$, see a review in \cite{Mourao1999}~Section~II.1, or \cite{Reed1974} and there the Minlos theorem. This means that $\gamma$ is supported on such tempered distributions which are not more irregular than the elements of such Hilbert subspace within $\Sch'$.
 \item \label{R:GaussSuppNonzero:ii} For $N=1$ the subspace of continuous functions within $\Sch'$ supports $\gamma$, see \cite{Glimm1987}~Theorem~A.4.4, or \cite{Velhinho2017}~Section~9~Example~2. That is, for $N=1$ the measure $\gamma$ is supported on fairly regular distributions, which are actually functions, and this explains why, in quantum mechanics renormalization is not needed.
 \item \label{R:GaussSuppNonzero:iii} For $N>1$ the subspace of those distributions in $\Sch'$ which on some open subset of $\R^{N}$ coincide with some signed measure has $\gamma$-measure zero, see \cite{Colella1973}~Proposition~3.1, or \cite{Velhinho2017}~Proposition~5. In particular, this applies to distributions in $\Sch'$ which coincide with some locally integrable function on some open subset of $\R^{N}$, even more specially to those which vanish on such patch. That is, $\gamma$ is supported on a subspace of fairly irregular tempered distributions.
 \item \label{R:GaussSuppNonzero:iv} For $N>1$, in \cite{Reed1974,Velhinho2017,Velhinho2015,Mourao1999} further details on the supporting subspaces and zero measure subspaces of $\gamma$ in $\Sch'$ are given. In particular, in \cite{Reed1974} it is shown that a supporting subspace of $\gamma$ can be found whose elements happen to be locally square integrable functions in a preferred direction, and are only genuine distributions of a certain kind in the remaining directions. That is, they retain some kind of regularity along a preferred direction.
 \item \label{R:GaussSuppNonzero:v} Let $\gamma$ and $\gamma'$ be centered Gaussian measures on $\Sch'$ such that the covariance of $\gamma$ majorates the covariance of $\gamma'$ by some Lipschitz constant. Then, every measurable linear subspace of full $\gamma$ measure is also of full $\gamma'$ measure (\cite{Bogachev1998}~Theorem~3.3.4).
 \item \label{R:GaussSuppNonzero:vi} Any two Gaussian measures $\gamma$ and $\gamma'$ on the same locally convex topological vector space $X$ are either mutually singular or mutually absolutely continuous to each-other (\cite{Bogachev1998}~Theorem~2.7.2). Consequently, if there exist measurable functions $f,f':\,X\to\Rbar_{0}^{+}$ such that $f\cdot\gamma=f'\cdot\gamma'$, and this product is not the zero measure, then $\gamma$ and $\gamma'$ are mutually absolutely continuous.
 \item \label{R:GaussSuppNonzero:vii} Two Gaussian measures on a locally convex topological vector space which are obtained from the same Gaussian measure (having non-finite dimensional support) via pushforward by field rescaling factors $z,z'\in\R^{+}$, are not mutually singular precisely when $z'=z$ (see e.g.\ \cite{Bogachev1998}~Example~2.7.4).
 \item \label{R:GaussSuppNonzero:viii} Two Gaussian measures on $\Sch'$ associated to the partial differential operators $(-d^{2}\Delta+m^{2})$ and $(-d'^{2}\Delta+m'^{2})$ ($d,d'>0$ and $m,m'>0$) are not mutually singular, precisely when $d'=d$ and $m'=m$ (consequence of \cite{Bogachev1998}~Corollary~6.4.8). The same is true for their pushforwards by $C_{\eta}$ ($\eta\in\Sch\setminus\{0\}$).
\end{enumerate}
\end{remark}

\section*{References}

\bibliographystyle{JHEP}
\bibliography{wilson}
\end{document}